\documentclass[prb, reprint, 9pt, superscriptaddress,
notitlepage, nofootinbib, longbibliography,
floatfix]{revtex4-2}

\usepackage{natbib}
\usepackage{graphicx}
\usepackage{mathtools}
\usepackage{amsfonts} 
\usepackage{amsmath} 
\usepackage{amssymb}
\usepackage{amsthm}
\usepackage{bm}
\usepackage{braket}
\usepackage{color} 
\usepackage{dsfont}
\usepackage[colorlinks]{hyperref}

\hypersetup{
    colorlinks,
    breaklinks,
    urlcolor=[rgb]{0.15,0.25,0.70},   %
    linkcolor=[rgb]{0.15,0.25,0.70},
    citecolor=[rgb]{0.15,0.25,0.70}
}

\usepackage[capitalize]{cleveref}

\usepackage{simplewick}
\usepackage{import}
\usepackage{microtype}
\usepackage{relsize}
\usepackage{enumitem}
\usepackage{stackengine}
\usepackage{wasysym}
\usepackage[dvipsnames]{xcolor}
\graphicspath{{images/}}

\usepackage{tikz}
\usepackage{pgfplots}
\pgfplotsset{compat=1.18}
\usetikzlibrary{calc}
\usetikzlibrary{arrows.meta} 
\usetikzlibrary{decorations.pathmorphing, fadings}
\usetikzlibrary{decorations.pathreplacing, intersections}

\usepackage[table]{xcolor}
\usepackage{makecell}
\usepackage{booktabs}
\usepackage{array}

\newcommand{\psym}{\mathcal{S}}
\newcommand{\tsym}{\mathcal{T}}

\usepackage{widebar}

\newcommand{\X}{\widebar{X}}

\newcommand{\kett}[1]{|#1\rangle\!\rangle}
\newcommand{\bbra}[1]{\langle\!\langle #1|}
\newcommand{\kettbbra}[2]{|#1\rangle\!\rangle\!\langle\!\langle #2|}
\newcommand{\bbrakett}[1]{\langle\!\langle #1\rangle\!\rangle}

\DeclareMathOperator{\uMPS}{uMPS}

\newcommand{\norm}[1]{\left\lVert#1\right\rVert}

\newcommand{\defeq}{\vcentcolon=}

\newcommand{\ketbra}[2]{\ket{#1}\!\bra{#2}}

\newcommand{\one}{\mathds{1}}
\newcommand{\R}{\mathbb{R}}
\newcommand{\Z}{\mathbb{Z}}
\newcommand{\N}{\mathbb{N}}
\newcommand{\C}{\mathbb{C}}

\newcommand{\Hilb}{\mathcal{H}}

\DeclareMathOperator{\Tr}{Tr}

\DeclareMathOperator{\spn}{span}

\newtheorem{theorem}{Theorem}

\newtheorem{corollary}{Corollary}
\newtheorem{lemma}{Lemma}
\newtheorem{proposition}{Proposition}

\theoremstyle{definition}

\usepackage{tcolorbox}

\makeatletter
\newsavebox{\@brx}
\newcommand{\llangle}[1][]{\savebox{\@brx}{\(\m@th{#1\langle}\)}%
  \mathopen{\copy\@brx\kern-0.5\wd\@brx\usebox{\@brx}}}
\newcommand{\rrangle}[1][]{\savebox{\@brx}{\(\m@th{#1\rangle}\)}%
  \mathclose{\copy\@brx\kern-0.5\wd\@brx\usebox{\@brx}}}
\makeatother

\newcommand{\tensorColor}{cyan!30}

\newcommand{\disctensor}[5]{
    \begin{scope}[shift={(#1,0)}]
        \draw[thick] (0, 0.3) -- (0, 0.7) node[above] {#3};
        
        \draw[thick, rounded corners=8pt] (-0.3, 0) -- (-0.8, 0) -- (-0.8, 0.7) node[above] {#2};
        
        \draw[thick, rounded corners=8pt] (0.3, 0) -- (0.8, 0) -- (0.8, 0.7) node[above] {#4};
        
        \filldraw[fill=\tensorColor, thick] (-0.3,-0.3) rectangle (0.3,0.3) node[midway] {#5};
    \end{scope}
}

\usepackage{pifont}

\newcommand{\myYes}{\textcolor{green!60!black}{\large\ding{51}}}
\newcommand{\myNo}{\textcolor{red!80!black}{\large\ding{55}}}

\usepackage{amsmath}
\usetikzlibrary{decorations.pathreplacing, arrows.meta, calc, intersections}

\definecolor{mypink}{RGB}{245, 163, 164}
\definecolor{mygreen}{RGB}{147, 204, 118}

\pgfdeclarelayer{background}
\pgfdeclarelayer{foreground}
\pgfsetlayers{background,main,foreground}

\tikzset{
    pinkdot/.style={circle, fill=mypink, inner sep=2pt},
    greendot/.style={circle, fill=mygreen, inner sep=2pt},
    cross/.style={preaction={draw, white, line width=4.5pt}},
    thickline/.style={thick, line width=1.5pt},
    reddotted/.style={very thick, red, densely dotted}
}

\def\xdistance{1.2}  %
\def\ydistance{0.8}  %

\newcommand{\drawblock}[5]{
    \begin{scope}[shift={(#4,#5)}]
        \foreach \i in {1,...,#2} {
            \coordinate (#1-T-\i) at (\i*\xdistance + #3, \ydistance);
            \coordinate (#1-B-\i) at (\i*\xdistance, 0);
            
            \begin{pgfonlayer}{foreground}
                \node[pinkdot] at (#1-T-\i) {};
                \node[greendot] at (#1-B-\i) {};
            \end{pgfonlayer}
        }
    \end{scope}
}

\newcommand{\connectIntra}[4][]{
    \draw[thickline, #1] (#2-B-#3) -- (#2-T-#4);
}

\newcommand{\connectInter}[7][]{
    \draw[thickline, #1] (#2-#3-#4) -- (#5-#6-#7);
}

\begin{document}
	
	\title{Mixed-State Long-Range Entanglement from Dimensional Constraints}
	
		\author{Leonardo A. Lessa}\email{llessa.physics@pm.me}
			\affiliation{Perimeter Institute for Theoretical Physics, Waterloo, Ontario N2L 2Y5, Canada}
			\affiliation{Department of Physics and Astronomy, University of Waterloo, Waterloo, Ontario N2L 3G1, Canada}

		\author{Tsung-Cheng Lu}\email{tclu@umd.edu}
			\affiliation{Joint Center for Quantum Information and Computer Science,
University of Maryland, College Park, Maryland 20742, USA}

	\begin{abstract}
    We present a new mechanism for long-range entanglement (LRE) in strongly symmetric many-body mixed states that does not rely on symmetry anomalies or long-range correlations. Our primary example is the maximally mixed state in the translation-invariant subspace on a one-dimensional ring. This state is LRE because translationally symmetric short-range entangled states span a subspace whose dimension grows only polynomially with system size, whereas the full translation-invariant subspace grows exponentially. We further discuss certain unconventional properties of this state, including logarithmically growing conditional mutual information, strong-to-weak spontaneous symmetry-breaking, and Rényi-index-dependent operator-space entanglement. We also construct a geometrically non-local Lindbladian to stabilize this state as the steady state. Our results identify dimensional mismatch as a novel route to LRE that is intrinsic to many-body mixed states.  
    \end{abstract}
	\maketitle
	
	{
		\hypersetup{linkcolor=black}
	}

One central goal in quantum many-body physics is the exploration of long-range entangled (LRE) quantum states, which underlie a variety of exotic phenomena such as topological orders and quantum criticality. While traditional studies have largely focused on isolated systems described by pure states, any realistic quantum system inevitably interacts with its environment and is therefore described by a mixed state. Then, a natural question arises: how should one define long-range entanglement for many-body mixed states? One commonly adopted definition is that a mixed state is LRE if it does not admit an ensemble of short-range entangled (SRE) pure states, each of which can be adiabatically connected to a product state \cite{Hastings_2011_topo}. Although this is a natural generalization from pure states to mixed states, such a definition makes certifying long-range entanglement challenging in general. In particular, since a mixed state admits an infinite number of ensemble realizations, how should one conclusively claim that the mixed state cannot admit an ensemble of SRE states? So far, there have been three mechanisms that guarantee long-range entanglement:

(i) \emph{Symmetry anomaly}: if a mixed state $\rho$ has a strong symmetry $U : G \to \mathcal{U}(\Hilb)$, that is, $U_g\rho = e^{i\theta} \rho$, then all pure states $\ket{\psi_\alpha}$ in any ensemble decomposition $\rho= \sum_\alpha p_\alpha \ketbra{\psi_\alpha }{\psi_\alpha}$ must also be symmetric with same charge: $ U_g \ket{\psi_\alpha} = e^{i \theta} \ket{\psi_\alpha}$. This is stronger than the usual \textit{weak} symmetry condition $U_g \rho = \rho U_g$, which does not guarantee all pure states in arbitrary ensembles are even symmetric. Therefore, if the strong symmetry $U$ is anomalous, which forbids any symmetric SRE pure state, then $\rho$ must be LRE. Notable examples include the 1D CZX mixed state \cite{lessa_2025_Multipartite, wang_anomaly_2024, ruiz-de-alarcon_matrix_2024a, sun_anomalous_2025, liu_establishing_2026}, the below-threshold 2d toric code \cite{grover_separability_2024,lessa_2025_higher,grover_2025_mixed_TEE}, and various intrinsically mixed-state topological orders \cite{Wang_2025_Intrinsic,Ellison_2025_mixed,Prem_2025_noisy,Ellison_2025_3d}.

(ii) \emph{Strong symmetry and long-range correlation}: if a mixed state $\rho$ with a strong symmetry $U$ can be an ensemble of $U$-symmetric SRE states, the Lieb-Robinson bound \cite{lieb_robinson_1972,hastings_LR_bound_2006} dictates that, for each SRE pure state $\ket{\psi_\alpha}$ in the ensemble, the connected correlation of charged local operators $O_i$ and $O_j^\dagger$, defined as $\langle O_iO_j ^{\dagger} \rangle_{\psi_\alpha}  -  \langle O_i \rangle_{\psi_\alpha}   \langle O_j ^{\dagger} \rangle_{\psi_\alpha}$, must decay exponentially with the separation between $i$ and $j$. Furthermore, the one-point function vanishes $\langle O_i \rangle_{\psi_\alpha}   =  \langle O_j ^{\dagger} \rangle_{\psi_\alpha} =0 $ due to the symmetry of $\ket{\psi_\alpha}$. It then follows that the two-point function w.r.t. the mixed state $\Tr [\rho O_i O_j^{\dagger}]$ also decays exponentially. Therefore, a strongly symmetric mixed state $\rho$ with long-range correlation (e.g. $\Tr[ \rho O_i O_j^{\dagger}]$ is constant) does not admit any SRE pure-state decomposition, thus implying LRE. Such constraints on mixed states was first employed in Ref.~\cite{Lu_2023_mixed_stateLRE} in the context of state preparation with measurement and feedback, and later in Refs.~\cite{grover_2024_symmetry_separable,Sahu_2024_mmis}. 

(iii) \emph{Large entanglement}: if a mixed state $\rho$ is a mixture of SRE pure states preparable by a finite-depth local unitary circuit, then the bipartite entanglement of $\rho$ is upper bounded by an area law\footnote{For instance, using entanglement of formation (Eq.~\ref{eq:eof}) as a measure, one can see that the entanglement entropy averaged over those SRE pure states provides an upper bound.}. Therefore, if the entanglement of $\rho$ scales faster than the area of a subregion, $\rho$ will necessarily be LRE. Examples include the volume-law entangled mixed states obtained by tracing out a sufficiently small subsystem from Haar random states or from highly excited eigenstates of non-integrable Hamiltonians \cite{AUBRUN_2012_partial_transpose,Aubrun_2012_semi_circle,Lakshminarayan_2012_wigner,Lu_2020_nega_transition,Shapourian_2021_negativity}.  

\begin{figure*}
    \centering
    
    \begin{minipage}[c]{0.26\textwidth}
        \centering
        \resizebox{0.95\linewidth}{!}{\begin{tikzpicture}[>=latex]

    \def\N{12}               %
    \def\R{3.0}              %
    \def\Renv{0.7*\R}
    \def\Rarrow{\R + 0.7}    %
    \def\atomradius{0.6}     %

    \def\Rcore{1.0}          %
    \colorlet{atomfill}{blue!20}
    \colorlet{envcolor}{red!30}
    \colorlet{linkcolor}{gray!70}
    \colorlet{nncolor}{blue!30!gray!85}

    \tikzset{
    atom/.style={circle, thick, ball color=blue!40, minimum size=\atomradius cm, inner sep=0pt},
    interaction/.style={thick, draw=linkcolor, decorate, decoration={snake, amplitude=0.6mm, segment length=3.5mm}},
    nn/.style={thick, draw=nncolor},
    tarrow/.style={-{Stealth[length=3mm, width=2.5mm]}, thick, draw=black!80}
    }

    \foreach \i in {1,...,\N} {
            \pgfmathsetmacro{\angle}{90 - (\i-1) * 360 / \N}
            \pgfmathsetmacro{\nexti}{int(mod(\i,\N)+1)}
            \pgfmathsetmacro{\nextangle}{90 - (\nexti-1) * 360 / \N}

            \coordinate (C\i) at (\angle:\R);
            \coordinate (Cnext) at (\nextangle:\R);

            \draw[interaction] (0,0) -- (C\i);

            \pgfmathsetmacro{\midangle}{\angle - 180/\N}
            \coordinate (C1_out) at (\angle : \R + 0.2);
            \coordinate (C1_in)  at (\angle : \R - 0.2);
            \coordinate (C2_out) at (\nextangle : \R + 0.2);
            \coordinate (C2_in)  at (\nextangle : \R - 0.2);

            \coordinate (M_curve_out) at (\midangle : \R + 0);
            \coordinate (M_curve_in)  at (\midangle : \R - 0.2);

            \coordinate (M12_out) at ($(C1_out)!0.5!(C2_out)$);
            \coordinate (M12_in)  at ($(C1_in)!0.5!(C2_in)$);
            \coordinate (Ctrl_out) at ($(M12_out)!2.0!(M_curve_out)$);
            \coordinate (Ctrl_in)  at ($(M12_in)!2.0!(M_curve_in)$);

            \filldraw[fill=nncolor, draw=nncolor!60!white, thick, line join=round] (C1_out) .. controls (Ctrl_out) .. (C2_out)
            -- (C2_in)  .. controls (Ctrl_in)  .. (C1_in) -- cycle;
        }

    \foreach \step in {0,...,40} {
            \pgfmathsetmacro{\r}{\Renv - \step * (\Renv-\Rcore)/40}
            \fill[envcolor, opacity=0.1] (0,0) circle (\r);
        }
    \fill[envcolor] (0,0) circle (\Rcore);
    \node[align=center, font=\LARGE\sffamily, text=red!70!black] (E) at (0,0) {Environment};

    \foreach \i in {1,...,\N} {
            \pgfmathsetmacro{\angle}{90 - (\i-1) * 360 / \N}
            \node[atom] (A\i) at (\angle:\R) {};
        }

    \pgfmathsetmacro{\startangle}{155}
    \pgfmathsetmacro{\endangle}{115}
    \draw[tarrow] (\startangle:\Rarrow) arc[start angle=\startangle, end angle=\endangle, radius=\Rarrow]
    node[midway, sloped, above=2pt, font=\Large] {$T$};

\end{tikzpicture}}
    \end{minipage}\hspace{1cm}
    \begin{minipage}[c]{0.55\textwidth}
        \centering
        \resizebox{\linewidth}{!}{   \begin{tikzpicture}[regionlabel/.style={font=\sffamily}]
    \colorlet{colorBackground}{orange!10}     %
    \colorlet{colorStateLabels}{YellowOrange!50!black} %
    \colorlet{colorTsymFill}{green!30!white}    %
    \colorlet{colorTsymBorder}{green}          %
    \colorlet{colorTsymLabel}{green!70!black}   %
    \colorlet{colorSREFill}{blue!35!white}      %
    \colorlet{colorSREBorder}{blue}             %
    \colorlet{colorSRELabel}{blue!60!black}     %

    \colorlet{colorBackground}{orange!15}
    \colorlet{colorStateLabels}{orange!60!black}

    \colorlet{colorTsymFill}{green!25!white}
    \colorlet{colorTsymBorder}{green!55!black}
    \colorlet{colorTsymLabel}{green!55!black}

    \colorlet{colorSREFill}{blue!15!white}
    \colorlet{colorSREBorder}{blue!70!black}
    \colorlet{colorSRELabel}{blue!55!black}

    \begin{scope}[yscale=0.9]
        \fill[colorBackground, rounded corners=0.2cm] (-4.5,-2) rectangle (4.5,2);
        \draw[colorStateLabels] (0,1.9) node[below, align=center, inner sep=4pt, regionlabel] {Pure quantum states};

        \begin{scope}[shift={(0,-0.3)}]
            \draw[fill=colorTsymFill, draw=colorTsymBorder, rounded corners=0.5cm] (-3,-1.3) rectangle (3,1.4);
            \draw[colorTsymLabel] (0, 1.2 ) node[below, align=center, scale=0.9, regionlabel] {Translationally invariant subspace};
            \draw[colorTsymLabel, <->, >=stealth, very thick] (3.2, 1.4) -- (3.2, -1.3) node[midway, right, font=\normalsize] {$\exp(L)$};

            \draw[fill=colorSREFill, draw=colorSREBorder, rounded corners=0.4cm] (-1.4,-0.9) rectangle (1.4,0.5);
            \draw[colorSRELabel] (0, 0.4) node[below, align=center, scale=0.9, regionlabel] {Span of SRE states};
            \draw[colorSRELabel, <->, >=stealth, thick] (1.6, 0.5) -- (1.6, -0.9) node[midway, right, font=\small] {$\text{poly}(L)$};
            
            \draw (0, -0.4) node[scale=0.8] {e.g. $\ket{\psi}^{\otimes L}$, $|\text{GHZ}\rangle$};
        \end{scope}
    \end{scope}

\end{tikzpicture}}
    \end{minipage}
    
    \vspace{0.2cm} %

\begin{minipage}[t]{0.26\textwidth}
    \centering
    {\small\textbf{(a)}}
\end{minipage}\hspace{1cm}
\begin{minipage}[t]{0.55\textwidth}
    \centering
    {\small\textbf{(b)}}
\end{minipage}

    \caption{\textbf{(a)} We explore the properties of the maximally mixed invariant state (MMIS) $\rho_\tsym$ of strong translation symmetry generated by the translation operator $T$.  $\rho_\tsym$  can be realized as the steady state of generic Lindblad dynamics that preserves $T$ symmetry \cite{zhang_stationary_2020, yoshida_uniqueness_2024, Sahu_2024_mmis}, e.g., by considering the neighboring interactions $H = \sum_{i} h_{i,i+1}$ and a uniform coupling to an environment via Lindbladian jump operator $L = \sum_i O_i$. \textbf{(b)} Schematic comparison between the translation-invariant subspace $\tsym$ and the span of SRE states within it, as a function of the system length $L$: the former contains all pure states with translation symmetry and the dimension of the space grows exponentially in $L$ (Lemma \ref{lemma:tsym_dim}). The latter contains all pure states that can be expressed as a linear combination of SRE pure states, including all permutation-symmetric states such as product states and the GHZ state; see Appendix.\ref{sec:span_product} for details. Its dimension is bounded by a polynomial in $L$ (Lemma \ref{lemma:uMPS_dim}), thus occupying a vanishing fraction of $\tsym$ as $L\to\infty$. This implies that the MMIS of translation is long-range entangled (Theorem \ref{thm:rho_LRE}).}
    \label{fig:optics_state_spaces}
\end{figure*}
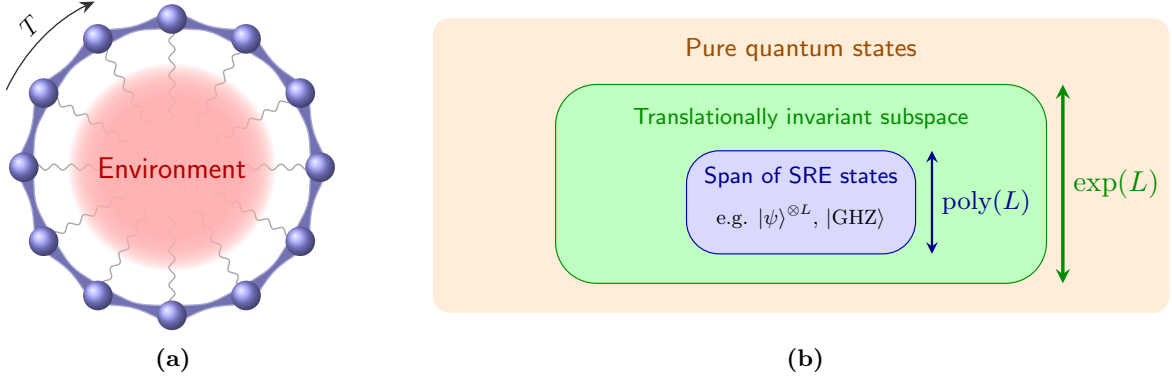

In this work, we identify a new mechanism that guarantees mixed-state long-range entanglement. As an illustration, we construct a class of LRE mixed states that exhibit none of the properties listed above. Namely, they have no symmetry anomaly, no long-range correlations of local operators, and no extensive bipartite entanglement. Our primary example is the \textit{maximally mixed invariant state} (MMIS)  $\rho_{\tsym}$ with strong lattice translation symmetry; see Fig.\ref{fig:optics_state_spaces}(a). $\rho_\tsym$ is defined on a 1d lattice of qudits with periodic boundary conditions as 
\begin{equation}\label{eq:intro_rho_tsym}
\rho_{\tsym} \propto  P_\tsym = \frac{1}{L}(\one+ T+ T^2+ \cdots + T^{L-1}),
\end{equation}
where $L$ is the number of lattice sites on the 1D ring, $T$ is the operator that implements the translation by one site, and $P_\tsym$ is the projector into the zero-momentum translation-invariant subspace. By defining the computational basis of $d$-dimensional qudits\footnote{Here, the on-site Hilbert space dimension $d$ is fixed and independent of the system size $L$.} $\ket{\boldsymbol{\sigma}} \equiv  \ket{\sigma_1, \sigma_2, ..., \sigma_L}$ with $\sigma_i  \in \{0,1,2,..., d-1 \}$, the mixed state $\rho_\tsym$ can be expressed as $\rho_{\tsym} \propto \sum_{\boldsymbol{\sigma} } P_\tsym \ketbra{\boldsymbol{\sigma}}{\boldsymbol{\sigma}} P_\tsym$. Namely,  $\rho_\tsym$ can be realized as an ensemble of pure states $P_\tsym \ket{\boldsymbol{\sigma}}$, each of which is a uniform superposition of all possible cyclic shifts of a computational basis state $\ket{\boldsymbol{\sigma}}$. Crucially, $T$ is \textit{not} anomalous, since it admits symmetric SRE states, such as any homogeneous product state $\ket{\psi}^{\otimes L }$ \footnote{In an alternative definition \cite{else_2026_anomaly}, $T$ is anomalous as it belongs to a blend equivalence class distinct from onsite symmetries. This definition implies that there is an obstruction for gauging $T$ symmetry.}. Yet, as we will discuss in this work, the MMIS $\rho_\tsym$ is LRE. We note that a $T$-symmetric pure state $\ket{\psi}$ can have a non-zero momentum, i.e. $T\ket{\psi} = e^{ik} \ket{\psi}$ with $k\neq 0 \pmod{2\pi}$, but such a non-zero momentum is anomalous in the sense that it forbids SRE pure states \cite{wang_2022_momentum}. For this reason, we restrict our discussion to translation symmetry in the zero-momentum sector.

The fundamental reason why $\rho_{\tsym}$ is LRE is that the dimension of the space spanned by $T$-invariant SRE states scales at most polynomially in $L$, while the dimension of the whole $T$-invariant subspace %
grows exponentially with $L$; see Fig.\ref{fig:optics_state_spaces}(b). As such, \textit{there are not enough} $T$-invariant SRE states to exhaust the whole $T$-invariant subspace, so $\rho_\tsym$  must be LRE. This \textit{dimensional constraint} provides a new mechanism for mixed-state LRE that is intrinsic to mixed states. In contrast, the mechanisms for mixed-state LRE studied in the literature so far (i.e., the three summarized earlier) are not intrinsic to mixed states since pure states can also be LRE via the same mechanisms. Relatedly, it is also insightful to compare translation symmetry with other symmetries and notice its unique pattern that is absent in the others (see \cref{tab:symmetry_comparison}). 

\begin{table}[t]
    \centering
    
    \footnotesize 
    \setlength{\tabcolsep}{1.7pt} 
    
    \setlength{\aboverulesep}{0pt} 
    \setlength{\belowrulesep}{0pt} 
    \setlength{\extrarowheight}{4pt} 
    \renewcommand{\arraystretch}{1.2} 
    
    \begin{tabular}{@{} l c c c >{\columncolor{cyan!8}[\tabcolsep][\tabcolsep]}c } 
         \toprule
         & \multicolumn{4}{c}{\raisebox{2.5pt}{\small\textbf{Symmetry}}} \\ 
         \cmidrule{2-5} 
         
         & \makecell{\textbf{Abelian} \\ \textbf{(on-site)}} 
         & \makecell{\textbf{Non-Abelian} \\ \textbf{continuous}} 
         & \makecell{\textbf{Anomalous}}
         & \textbf{Translation} \\ 
         \midrule
         
         \makecell[l]{All symm. \\ states LRE } & \myNo & $\text{\myNo}$ & \myYes & \myNo \\
         
         \makecell[l]{MMIS LRE} & \myNo & \myYes & \myYes & \myYes \\

         \makecell[l]{MMIS LRC} & \myNo & $\text{\ \myYes}^*$ & $\ \text{\myNo}^\dagger$ & \myNo \\ 
         \bottomrule
    \end{tabular}
    
    \caption{
    Comparison of long-range entanglement (LRE) and long-range correlation (LRC) of local observables for various symmetry classes.
    ${\!}^*$~Because two-point correlations $\langle O_i O_j \rangle$ of the MMIS vanish algebraically with system size, and not with distance $|i-j|$, it evades Lieb-Robinson bounds~\cite{Sahu_2024_mmis, li_highly_2025}. ${\!}^\dagger$~For a MMIS with global (0-form) anomalous symmetry ~\cite{lessa_2025_Multipartite,lessa_2025_higher}.}
    \label{tab:symmetry_comparison}
\end{table}

In the rest of the manuscript, we prove that $\rho_{\tsym}$ exhibits long-range entanglement, analyze its correlation properties and information-theoretic measures, and conclude with possible generalizations and future directions. As a warm-up, we first present a theorem regarding the on-site non-separability of $\rho_{\tsym}$, before treating its LRE properties:

\begin{theorem}\label{thm:tsym_entangled}
    The translation symmetry MMIS $\rho_\tsym$ is entangled (i.e., not fully separable) for large enough $L$. Specifically, $\rho_\tsym  \neq  \sum_{\alpha} p_\alpha (\ketbra{\psi_{\alpha}}{\psi_{\alpha}})^{\otimes L}$, for $L \geq 3$ when $d \geq 3$, and for $L \geq 4$ when $d=2$, and is fully separable otherwise.
\end{theorem}
We prove this theorem by contradiction: assuming the contrary, then there exists a separable decomposition: $\rho_\tsym = \sum_\alpha p_{\alpha}  ( \ketbra{\psi_\alpha}{\psi_\alpha}  )^{\otimes L}$. For that, the set of states in the decomposition has to span the whole $\tsym$ subspace. However, the product state $\ket{\psi_\alpha}^{\otimes L}$ can at most span a \textit{permutationally} symmetric subspace $\psym \subseteq \tsym$ (see Appendix.\ref{sec:span_product}). This leads to a contradiction for large enough $L$, as $\dim(\psym) = \Theta(L^{d-1})$ grows only polynomially, whereas $\dim(\tsym) = \Theta(d^L/L)$ grows exponentially with $L$, based on the following two Lemmas (see Appendix. \ref{proof_lemma:psym_dim} and \ref{proof_lemma:tsym_dim} for proofs):

\begin{lemma}\label{lemma:psym_dim}
    The dimension of the permutationally symmetric subspace $\psym \subseteq (\C^d)^{\otimes L}$ of $L$ qudits is 
    \begin{equation}
        \dim(\psym) = \binom{L+d-1}{d-1}.
    \end{equation}
 To leading order in $L$, $\dim(\psym) $ scales as $\frac{L^{d-1}}{(d-1)!}$.
\end{lemma}

\begin{lemma}\label{lemma:tsym_dim}
    The dimension of the translationally symmetric subspace $\tsym \subseteq (\C^d)^{\otimes L}$ of $L$ qudits is 
    \begin{equation}
        \dim(\tsym) = \frac{1}{L} \sum_{i=0}^{L-1} d^{\gcd(i,L)}.
    \end{equation}
    To leading order in $L$,  $\dim(\tsym)$ scales as $d^L/L$, up to a correction that is exponentially small in L.
\end{lemma}

For system sizes $L$ satisfying the conditions stated in Theorem \ref{thm:tsym_entangled}, one finds $ \dim(\psym) < \dim(\tsym )$ (see Appendix.~\ref{appendix:table_dim} for detailed comparison), and therefore $\rho_\tsym$ is entangled. Otherwise, one has $\psym = \tsym$ and $\rho_\tsym = \rho_\psym$, and we prove that the MMIS $\rho_\psym$ is fully separable in Appendix.~\ref{sec:span_product}.

From Theorem \ref{thm:tsym_entangled}, we see that $\rho_\tsym$ is not fully separable for a large enough $L$. In fact, the strong translation symmetry allows us to prove a stronger theorem:

\begin{theorem}\label{thm:tsym_bipent}
    The translation symmetry MMIS $\rho_\tsym$ is bipartite entangled for large enough $L$. Specifically, $\rho_\tsym  \neq  \sum_{\alpha} p_\alpha  \ketbra{A_\alpha}{ A_{\alpha} } \otimes\ketbra{B_\alpha}{ B_{\alpha} } $, for $L \geq 3$ for $d \geq 3$, and $L \geq 4$ for $d=2$.
\end{theorem}

This theorem follows from the following Lemma: 

\begin{lemma}\label{lemma:tsym_bipsep_implies_fullsep}
    A translationally symmetric pure state $\ket{\Psi} \in \Hilb^{\otimes L}$ (i.e. $T \ket{\Psi} \propto \ket{\Psi}$) that is bipartite separable with respect to contiguous regions $A|B$, i.e. $\ket{\Psi} = \ket{A}\ket{B}$, is fully separable: $\ket{\Psi} = \ket{\psi}^{\otimes L}$.
\end{lemma}
The central proof idea for this Lemma is to note that when $\ket{\Psi} = \ket{A}\ket{B}$, one can apply a lattice translation to construct various overlapping regions that satisfy the separability condition as well. This, in turn, eliminates all correlations, hence yielding a fully separable product state (see Appendix.~\ref{append:bipartite_fully_separable}). With Lemma \ref{lemma:tsym_bipsep_implies_fullsep} above, Theorem \ref{thm:tsym_bipent} reduces to Theorem \ref{thm:tsym_entangled}.

So far, Theorem \ref{thm:tsym_entangled} and Theorem \ref{thm:tsym_bipent} dictate that $\rho_\tsym$ is not fully separable, and is in fact, bipartite entangled. A natural question is whether such entanglement is short-range or long-range. In the following, we will address this question within the matrix product state (MPS) formalism. First, we define an uniform MPS (uMPS) as
\begin{equation}
\ket{\psi (A)} = \sum_{\{\sigma_i \}} \Tr[A^{\sigma_1 }  A^{\sigma_2 } \cdots A^{\sigma_L } ] \ket{\sigma_1 \sigma_2 \cdots \sigma_L}.
\end{equation}
where $A^{\sigma}$ is a $\chi \times \chi $ matrix, with $\chi$ known as the bond dimension. The tensor $A$ does not depend on the site label. A uMPS can also be pictorially represented as 
\begin{equation}\label{eq:uMPS_picture}
\begin{aligned}
\begin{tikzpicture}[thick]%
        \filldraw[fill = \tensorColor] (1.7,-.3) rectangle node {$A$} (2.3,.3);
        \draw (2,.3) -- (2,.6) node[above] {$\sigma_1$};
        \draw (2.7,0) --  (2.3,0);
        
        \filldraw[fill = \tensorColor] (2.7,-.3) rectangle node {$A$} (3.3,.3);
        \draw (3,.3) -- (3,.6) node[above] {$\sigma_2$};
        \draw (3.5,0) --  (3.3,0);
        
        \node at (3.8,0) {\ldots};
        \draw (4.1,0) --  (4.3,0);
        
        \filldraw[fill = \tensorColor] (4.3,-.3) rectangle node {$A$} (4.9,.3);
        \draw (4.6,.3) -- (4.6,.6) node[above] {$\sigma_L$};
        
        \draw[rounded corners=8pt] (1.7,0) -- (1.2,0) -- (1.2,-0.6) -- (3.5,-0.6);
        \node at (3.8,-0.6) {\ldots};
        \draw[rounded corners=8pt] (4.1, -0.6) -- (5.4,-0.6) -- (5.4,0) -- (4.9,0);
\end{tikzpicture}
\end{aligned}  
\end{equation}

Within the MPS formalism, we prove 
\begin{theorem}\label{thm:rho_LRE}
    The translation symmetry MMIS $\rho_\tsym  \neq   \sum_{\alpha} p_{\alpha}\ketbra{\psi_\alpha}{\psi_\alpha}$, where $\ket{\psi_\alpha}$ is a constant bond-dimension uMPS for large enough $L$.    
\end{theorem}

To appreciate how this theorem constrains the entanglement structure of $\rho_\tsym$, we first recall that in the context of MPSs, $\ket{\psi(A)}$ is called short-range entangled (SRE) \cite{cirac_mps_review} if it is a normal MPS with finite bond dimension\footnote{An MPS is normal if and only if its MPS tensor, under blocking a finite number of sites, defines an injective map from the virtual space to the physical space \cite{cirac_mps_review}. This is in contrast to the long-range entangled GHZ state, whose MPS representation is not normal.}. Normal MPSs have several important properties: they arise as unique ground states of a local gapped Hamiltonian, and their transfer matrices have a unique largest eigenvalue, implying a finite spectral gap and a finite correlation length. Crucially, every translationally invariant normal MPS admits a uMPS representation \cite{cirac_mps_review}. One therefore has the following result:  

\begin{corollary}\label{corollary:lre}
The translation symmetry MMIS is long-range entangled for large enough $L$. Specifically, $\rho_\tsym  \neq   \sum_{\alpha} p_{\alpha}\ketbra{\psi_\alpha}{\psi_\alpha}$, where $\ket{\psi_\alpha}$ is an SRE state represented by uMPS.
\end{corollary}

We note that Theorem \ref{thm:rho_LRE} imposes a stronger constraint on the structure of the mixed state $\rho_\tsym$ than Corollary \ref{corollary:lre}. This is because uMPS with constant bond dimension form a broader class of states than just short-range entangled uMPS, including for example GHZ-type states.

The proof for Theorem \ref{thm:rho_LRE} follows from a dimensional counting based on the following lemma:
\begin{lemma}
\label{lemma:uMPS_dim}
The dimension of the linear span of all uMPSs is upper bounded by
    \begin{equation}
        \dim(\spn(\uMPS(\chi, d, L))) \leq \binom{L+d\chi^2-1}{d\chi^2-1},
    \end{equation}
     where $\uMPS(\chi, d, L)$ is the set of uMPS of bond dimension $\chi$ over $L$ $d$-dimensional sites. In particular, it can grow at most polynomially as $O(L^{d\chi^2-1})$ as $L \to \infty$.
\end{lemma}

This Lemma first appeared in Ref.~\cite{werner_2006_mps_dimension}; see also Ref.\cite{de2022linear} for a related discussion. Here we provide an alternative and intuitive derivation: 
\begin{proof}
    Given a local tensor $A_{ij}^{\sigma}$, we can construct a product state $\ket{\Psi(A)} \in (\C^{d\chi^2})^{\otimes L}$ in an artificial Hilbert space of a chain of sites with dimension $d\chi^2$ by combining virtual and physical degrees of freedom of the MPS. In the local basis $\{(\sigma, i, j)\}$ of $\C^{d \chi^2} \simeq \C^{d} \otimes \C^\chi \otimes \C^\chi$, the wavefunction is $\braket{(\sigma_1, i_1,j_1), \ldots, (\sigma_L, i_L, j_L) |\Psi(A)} \defeq \mathcal{N{{}}} A_{i_{1}j_{1}}^{\sigma_1} \cdots A_{i_{L}j_{L}}^{\sigma_L}$, where $\mathcal{N}$ is a normalization constant. Diagrammatically, this wavefunction can be represented as
    \vspace{-0.5em}
    \begin{align}
            \begin{tikzpicture}[baseline=10pt]
                \disctensor{0}{$i_{1}$}{$\sigma_1$}{$j_{1}$}{$A$}
                \disctensor{3.5}{$i_{L}$}{$\sigma_L$}{$j_{L}$}{$A$}
                \node at (1.75, 0.4) {\LARGE \dots};
            \end{tikzpicture}
    \end{align}

    Each of the states $\ket{\Psi(A)}$ is symmetric under arbitrary permutations of the $d \chi^2$-dimensional sites. Thus, by Lemma \ref{lemma:psym_dim}, the subspace $V$ spanned by the product state $\ket{\Psi(A)} \in (\C^{d\chi^2})^{\otimes L} $ has a dimension $\dim(V)  = \binom{L + d\chi^2 - 1}{d\chi^2 -1}$.

    Furthermore, by projecting all neighbouring virtual legs of $\ket{\Psi(A)}$ onto the Bell state $\ket{00}+\ket{11}$, we can glue all virtual legs to construct the MPS wavefunction $\Tr[A^{\sigma_1} A^{\sigma_2} A^{\sigma_3} \cdots]$ that is pictorially represented as
    \begin{equation}
        \makebox[0.85\linewidth][c]{
        \begin{tikzpicture}[baseline=(current bounding box.center)]
            \foreach \x in {1, 2, 3}
                \disctensor{{2.4*(\x-1)}}{}{$\sigma_\x$}{}{$A$}
            \draw[thick, rounded corners=8pt] (-1.2, 1.0) node[left] {\dots} -- (-0.8, 1.0) -- (-0.8, 0.7);
            \draw[thick, rounded corners=8pt] (0.8, 0.7) -- (0.8, 1.0) -- (1.6, 1.0) -- (1.6, 0.7);
            \draw[thick, rounded corners=8pt] (3.2, 0.7) -- (3.2, 1.0) -- (4.0, 1.0) -- (4.0, 0.7);
            \draw[thick, rounded corners=8pt] (5.6, 0.7) -- (5.6, 1.0) -- (6.0, 1.0) node[right] {\dots};
        \end{tikzpicture}
        }
    \end{equation}
    Since this projection cannot increase the subspace dimension, one finds 
    \begin{equation}
        \dim(\spn(\uMPS(\chi, d, L))) \leq \dim(V) = \binom{L+d\chi^2-1}{d\chi^2-1}
    \end{equation}
\end{proof}

In particular, any incoherent mixture of $\uMPS(\chi, d, L)$ can only form a mixed-state density matrix with a rank upper bounded by $  \dim(\spn(\uMPS(\chi, d, L)))$, which at most grows polynomially with $L$ for finite bond dimension $\chi$. On the other hand, the rank of MMIS $\rho_\tsym$ is the dimension of the translationally invariant subspace, which grows exponentially with $L$ (Lemma \ref{lemma:tsym_dim}). This dimensional mismatch therefore proves Theorem \ref{thm:rho_LRE}. 

As a final remark, we can even consider uMPS with a different virtual boundary condition, which is dubbed MPS-X in Ref.~\cite{mps-x_2025_cirac}. This consideration amounts to inserting a single tensor defect X in the virtual space of a uMPS defined under periodic boundary conditions: $\ket{\psi(A;X)} = \sum_{\{\sigma_i \}} \Tr[XA^{\sigma_1 }  \cdots A^{\sigma_L } ] \ket{\sigma_1 \cdots \sigma_L}$. Notably, MPS-X of finite bond dimension can represent the W state exactly, which would require a bond dimension growing polynomially with $L$ in the conventional uMPS. Based on our analysis above, the dimension of the span of MPS-X remains growing polynomially, so $\rho_\tsym$ cannot be a mixture of MPS-X with finite bond dimension as well.

Now using the strong translational symmetry, one can prove a stronger theorem to show that the long-range entanglement is bipartite in nature:

\begin{theorem}\label{thm:Bipartite_LRE}
    The translation symmetry MMIS $\rho_\tsym$ is long-range bipartite entangled for large enough $L$. That is, it is not a mixture of states of the form $U \ket{A}\ket{B}$, with $U$ a finite-depth local unitary circuit (FDLU). 
\end{theorem}

In order to prove this theorem, we introduce Lemma \ref{lemma:tsym_FDLU_bipsep_implies_fullsep} below, which relaxes Lemma \ref{lemma:tsym_bipsep_implies_fullsep} by allowing for short-range entanglement (see Appendix.~\ref{proof:tsym_FDLU_bipsep_implies_fullsep} for proof): 

\begin{lemma}\label{lemma:tsym_FDLU_bipsep_implies_fullsep}
    A translationally symmetric pure state $\ket{\Psi} \in \Hilb^{\otimes L}$ that is short-range bipartite entangled, i.e. $\ket{\Psi} = U \ket{A}\ket{B}$, with $U$ an FDLU and regions $A$ and $B$ large enough, 
    is a short-range entangled state preparable from a product state using an FDLU. 
\end{lemma}

With Lemma \ref{lemma:tsym_FDLU_bipsep_implies_fullsep} established, Theorem \ref{thm:Bipartite_LRE} directly follows from Theorem \ref{thm:rho_LRE}.

Finally, we note that the mechanism for long-range entanglement resulting from the dimensional constraint can already lead to many other examples beyond the MMIS $\rho_\tsym \propto P_\tsym$. For instance, for qubit systems ($d=2$), the state $ \rho \propto (1+ \prod_i X_i)P_\tsym$ with both spin flip and translation strong symmetries remains LRE since the matrix rank still grows exponentially with $L$. More generally, one can impose any finite symmetry (e.g., global $\Z_n$ or spatial reflection) to arrive at the same conclusion. Moreover, in higher spatial dimensions, similar arguments imply that the MMIS with translation symmetry along any spatial direction is not a mixture of uniform projected entangled pair states (PEPSs).

\textbf{Correlation structure.} So far, we have shown that $\rho_\tsym$ is long-range entangled (LRE), but its precise correlation and entanglement structure have yet to be determined. Indeed, LRE states can exhibit very different behaviors: they can exhibit volume-law entanglement scaling as in Haar random states, or appear as area-law entangled ground states of local Hamiltonians (as in the GHZ state, topological orders, or quantum critical states), with varying correlation and entanglement structures. In the following, we will present results on the correlation structures of $\rho_{\tsym}$. Detailed calculations and proofs can be found in Appendix.~\ref{appendix:correlation}. 

As a warm-up, we first notice that $\rho_{\tsym}$ has a weak symmetry under $U= u^{\otimes L}$, with $u$ being any single-site unitary operator:  $U \rho_{\tsym} U^{\dagger}  = \rho_{\tsym}$. This constrains the correlation structure of $\rho_\tsym$. For instance, consider an operator of the form $O_{a,b}=\prod_{i=1}^{L} X_i^{a_i} Z_{i}^{b_i}$, where $X_i$ and $Z_i$ are the generalized Pauli operators acting on $i$-th qudit\footnote{ $X$ and $Z$ are defined as follows: $X\ket{q}= \ket{q+1 ~\text{mod} ~d} $ and $Z\ket{q}= e^{\frac{2\pi iq}{d } }\ket{q}$, with $q= 0,1,2, ...  d-1$ labeling the computational basis of a $d$-dimensional qudit.}, and $a_i,b_i\in \{ 0,1,2,... d-1 \}$ for all $i$. The weak symmetry generated by $\prod_{i=1}^L X_i$ implies that $\langle O_{a,b} \rangle$ vanishes when $\sum_{i=1}^{L} b_i \neq 0 $ mod $d$. Similarly, the weak symmetry generated by $\prod_{i=1}^L Z_i$ implies that $\langle O_{a,b} \rangle$ vanishes when $\sum_{i=1}^{L} a_i \neq 0 $ mod $d$. Namely, any operator charged under the weak symmetry has a vanishing expectation value. In particular, this implies that the single-site reduced density matrix must be a maximally mixed state. More generally, we prove the following theorem: 

\begin{theorem}\label{thm:rdm}
Given the translation symmetry MMIS $\rho_\tsym$, the reduced density matrix $\Tr_{A^C}[\rho_\tsym]$ on any subset of sites $A$ is exponentially close to the maximally mixed state:
\begin{equation}
\norm{\Tr_{A^C}[\rho_\tsym] - \frac{\one}{d^{|A|}}}_1  \leq  O(Ld^{-(L-|A|)/2}),
\end{equation}
where $L$ is the total system size, and $|A|$ is the number of sites contained in $A$. In particular, when fixing the subregion size fraction $f= \frac{|A|}{L}$ and taking the thermodynamic limit $L\to\infty$, the local reduced density matrix is exactly a maximally mixed state for any $f<1$. 
\end{theorem}

Physically, this shows that the entanglement and correlations are encoded only nonlocally in the global structure of the mixed state and are invisible to any local subsystem. For example, with the methods described in Appendix.~\ref{appendix:correlation}, we find the expectation value of the two-point Pauli-Z operator $Z_i Z_j^{-1}$ for $L$ prime to be
\begin{equation}
       \forall i\neq j, \ \Tr [\rho_{\tsym} Z_i Z_j^{-1} ]   = \frac{d (L-1)  }{  d^L  + d (L-1) },   
\end{equation}
which decays exponentially with $L$. In fact, one can show that $\Tr [\rho_{\tsym}  \prod_{i=1}^L Z_i^{b_i}]$ takes the value given on the right-hand side of the equation above whenever $\sum_{i=1}^{L} b_i =0 ~\text{mod}~d$.

\textbf{Information-theoretic measures}. After presenting the correlation structure of $\rho_\tsym$, we now explore the properties of the state using information-theoretic measures. To quantify the bipartite entanglement in such a mixed state, one natural measure is the entanglement of formation \cite{eof_1996}, defined as the averaged entanglement entropy of a region $A$ minimized over all possible pure-state decompositions of the mixed state $\rho$:  
\begin{equation}\label{eq:eof}
E_f(\rho) =  \min_{ \{p_\alpha, \ket{\psi_\alpha}\}_\rho }\sum_\alpha p_\alpha  S_A ( \ket{\psi_\alpha}),
\end{equation}
where $\{p_\alpha, \ket{\psi_\alpha}\}_\rho$ denotes the set of decompositions of $\rho = \sum_\alpha p_\alpha \ketbra{\psi_\alpha}{\psi_\alpha}$. While we are not able to provide a precise estimate of $E_f (\rho_\tsym)$, it is straightforward to derive an upper bound:
\begin{theorem}\label{thm:eof}
The entanglement of formation $E_F$ of $\rho_\tsym$ w.r.t. any bipartition is upper bounded by $\log L$. 
\end{theorem}

To prove this theorem, consider the decomposition introduced right after \cref{eq:intro_rho_tsym}, $\rho_{\tsym} \propto \sum_{\boldsymbol{\sigma} } P_\tsym \ketbra{\boldsymbol{\sigma}}{\boldsymbol{\sigma}} P_\tsym$. For each $P_\tsym\ket{\boldsymbol{\sigma}}$, the reduced density matrix has a rank bounded by $L$, so that the bipartite entanglement entropy of $P_\tsym\ket{\boldsymbol{\sigma}}$ is bounded by $\log L$. This in turn implies the $\log L $ upper bound for the entanglement of formation.

We contrast this bound with the average entanglement of Haar random pure states with translation symmetry, which was found to follow volume-law bipartite entanglement~\cite{nakata_generic_2020}, instead of the much lower $O(\log L)$ from the decomposition above.

Another quantity of interest is the conditional mutual information (CMI) $I(A,C|B)$, which quantifies the information shared between two disconnected regions $A$ and $C$ mediated by the intermediate region $B$. It plays an essential role in the characterization of mixed-state phases of matter \cite{hsieh_markov_2025, sang_mixedstate_2025}. CMI is defined as $I(A,C|B) = S_{AB} + S_{BC} -S_B - S_{ABC}$, where $S_X$ denotes the von-Neumann entropy of the reduced density matrix on $X$. When $A, B, C$ all have a finite system fraction as $L\to  \infty$, Theorem \ref{thm:rdm} implies that the reduced states on $A\cup B$, $B\cup C$, and $B$ are simply the maximally mixed state on the corresponding region. It follows that 
\begin{equation}
I(A,C|B) =  \log L    
\end{equation}    
for $L \to \infty$. Physically, the non-vanishing CMI indicates the existence of certain global correlations that cannot be locally reconstructed (i.e., non-Markovian). Here, this global constraint comes from the strong translation symmetry.

Finally, we discuss the operator space entanglement, which captures the amount of (quantum and classical) correlations across a bipartition \cite{Zanardi_2001_operator,Zanardi_2002_operator,prosen_2007_operator_ising,prosen_2007_operator_xy,Dubail_2017}. It is closely related to the computational cost for simulating the mixed states using matrix product operators (MPO) \cite{2026_mpo_simulation}. To this end, we first vectorize the mixed state  $\rho_\tsym \propto \sum_{\boldsymbol{\sigma}}  P_\tsym \ketbra{\boldsymbol{\sigma}}{\boldsymbol{\sigma}}$ as a state defined in a doubled Hilbert space:
\begin{align} \label{eq:doubled}
\kett{\rho_{\tsym}} & \propto \sum_{\boldsymbol{\sigma}} P_\tsym \otimes \one \ket{\boldsymbol{\sigma}} \ket{\boldsymbol{\sigma}} \\
& \propto \sum_{\boldsymbol{\sigma}} (\one+T + ...+ T^{L-1}) \otimes \one \ket{\boldsymbol{\sigma}} \ket{\boldsymbol{\sigma}}.
\end{align}
where the first ket and the second ket denote the state in $\mathcal{H}^{\otimes L}$ and $\mathcal{\overline{H}}^{\otimes L}$ respectively. In other words, it is a superposition of Bell-pair configurations across two copies of the 1D ring, with $L$ distinct patterns of shifts.

By choosing a subsystem $A \cup \overline{A}$ consisting of the sites $\{ 1, 2, \cdots,  |A|\} \cup   \{ \overline{1}, \overline{2}, \cdots, \overline{  |A|  }\}$, the Rényi-$k$ operator space entanglement of $\rho_\tsym$ is given by the Rényi-$k$ entropy of the reduced density matrix on $A\cup \overline{A}$: $S_k = \frac{1}{1-k} \log \Tr \rho_{A \cup \overline{A}}^{k }$. In Appendix.~\ref{appendix:operator_space_entanglement}, we analytically find that for $\log L  \ll  |A| \leq \frac{L}{2}$, the von-Neumann entropy takes the form

\begin{equation}
S_{k=1}= 2 |A|\left(1-\frac{|A|}{L}\right) \log d + 2\frac{|A|}{L}\log L + O(1).
\end{equation}
Namely, the leading term follows a volume-law scaling, and interestingly, with a coefficient depending on the subregion fraction. This volume-law scaling implies that this doubled state cannot be efficiently simulated using MPS techniques \cite{mps_entropy_cost_2008}. Moreover, we find that all other Rényi entropies can be analytically calculated as well, under the same condition for $|A|$ specified above. For $k>1$, $S_k$ saturates to $\frac{k}{k-1} \log L$, i.e., a quantity that is a constant of $|A|$, but grows with the total system size $L$. For $k<1$, $S_k$ grows maximally as $2 |A| \log d$. 

\textbf{SW-SSB of lattice translation.} Strongly symmetric many-body mixed states can exhibit a novel phenomenon without pure-state counterparts known as the strong-to-weak spontaneous symmetry breaking (SW-SSB) \cite{lee_2023_swssb,you_2024_swssb,lessa_2025_swssb}, in which the strong symmetry can be spontaneously broken down to the weak one. The canonical example is the strong $\mathbb{Z}_2$ symmetric maximally mixed state $ \rho\propto \one + \prod_i X_i$, which, despite lacking true long-range order $\lim_{|i-j | \to \infty } \Tr [\rho Z_i Z_j]=0 $, it exhibits long-range order for certain non-linear observables such as the Rényi-2 correlator $R^{(2)} =\bbrakett{ \rho | Z_i \overline{Z}_i  Z_j \overline{Z}_j    | \rho  } =1 $, where $\kett{\rho} $ is the doubled state (associated with $\rho$) that has been properly normalized as $\bbrakett{\rho|\rho} = 1$. 

For simplicity, we choose a particular non-local operator $O_k= \frac{1}{L} \sum_{x=0}^{L-1} e^{ikx} Z_x$ that is charged under the lattice translation, $T^\dagger O_k T = e^{ik} O_k$. Then, we introduce the following \textit{variance-normalized} Rényi-2 correlator defined as
\begin{equation}
\tilde{R}_k^{(2)} =  \frac{  \bbra{\rho_\tsym}  O_k \overline{O}_{-k}  O_{-k} \overline{O}_{k}\kett{ \rho_\tsym}   }{  \sqrt{\bbrakett{   \rho_\tsym | (O_k \overline{O}_{-k}  O_{-k} \overline{O}_{k})^2| \rho_\tsym   }  }     }.
\end{equation}

This variance-normalized correlator $\tilde{R}_k^{(2)}$ is more appropriate to diagnose the SSB of translation due to the following reasoning: an SSB order parameter can be thought of as a probe for a response w.r.t. inserting symmetry charges, and this is exactly captured by $\tilde{R}_k^{(2)}$, which is the wavefunction overlap between two wavefunctions with and without the insertion of charged operators. Importantly, computing this overlap requires appropriate normalization given the non-unitary and non-local nature of $O_k$. This additional normalization factor distinguishes $\tilde{R}_k^{(2)}$ from the Rényi-2 correlator $R^{(2)}$ commonly defined in the literature \cite{lee_2023_swssb, lessa_2025_swssb, you_2024_swssb}. For most previously considered cases, this distinction does not change whether a state has SW-SSB or not, since either the charged operator is unitary or the difference in normalization amounts to a constant. For the case of the MMIS $\rho_\tsym$, however, this difference matters: as shown in Appendix. \ref{appendix:main_swssb}, its variance-normalized correlator $\tilde{R}_k^{(2)}$ is a non-zero constant, implying SW-SSB of the translation symmetry, despite $R^{(2)}$ vanishing in the thermodynamic limit. This also makes $\rho_\tsym$ follow the same SW-SSB pattern previously observed in MMISs associated with other symmetries, such as the $\one + \prod_i X_i$ state.

\textbf{State preparation.} Finally, we show that the MMIS of translation $\rho_\tsym$ can naturally emerge under a simple Markovian dynamics that preserves strong translation symmetry. We consider the Lindbladian $\mathcal{L}(\rho) = -i [H, \rho] + \sum_i \gamma_i (L_i \rho L_i^\dagger - \frac{1}{2} \{L_i^\dagger L_i, \rho\})$ of a one-dimensional ring of qubits generated by the transverse-field Ising Hamiltonian $H$, and a collective two-body jump operator $L_{XY}$:
\begin{equation}\label{eq:linblad}
\begin{split}
H & = - \sum_i Z_i Z_{i+1} - \lambda \sum_i X_i \\
    L_{XY} & = \sum_{i} X_i Y_{i+1}. 
\end{split}
\end{equation}

\begin{figure}[t]
    \centering
    \begin{tikzpicture}

\begin{axis}[
    width=\linewidth,
    height=5cm,
    xlabel={$t$},
    ylabel={$F(\rho(t), \rho_\mathcal{T})$},
    xlabel style={font=\normalsize},
    ylabel style={font=\normalsize},
    xmin=-0.5, xmax=10.5,
    ymin=0.15, ymax=1.05,
    grid=major,
    grid style={gray!20},
    axis background/.style={fill=white},
    legend pos=south east,
    legend cell align={left},
    legend style={nodes={scale=0.9, transform shape}},
    tick label style={font=\small},
    xtick={0, 2.5, 5.0, 7.5, 10.0},
    ytick={0.2, 0.4, 0.6, 0.8, 1.0},
    axis lines=left,
    axis line style={-},
    thick,
]

\addplot[
    color=cyan,
    line width=1.5pt,
] table[x index=0, y index=1] {sym_xy_data.txt};

\draw[dashed, cyan, line width=1.2pt] (axis cs:-0.5,1) -- (axis cs:10.5,1);

\coordinate (insetPosition) at (axis cs: 4.5, 0.4);

\end{axis}

\begin{axis}[
    at={(insetPosition)},
    anchor=south west,
    width=0.53\linewidth,
    height=3.3cm,
    ymode=log,
    ylabel={$1 - F$},
    ylabel style={font=\small},
    tick label style={font=\footnotesize},
    grid=major,
    grid style={gray!20},
    axis background/.style={fill=white},
    xmin=2.5, xmax=10.2,
    xtick={4,6,8,10},
    ymin=1e-5, ymax=1e-1,
    axis lines=left,
    axis line style={-}
]

\addplot[
    color=cyan,
    line width=1.5pt,
] table[x index=0, y expr={1 - \thisrowno{1}}] {sym_xy_data.txt};

\end{axis}

\end{tikzpicture}
    \caption{Fidelity between the MMIS $\rho_\tsym$ and the evolved state $\rho(t) = e^{\mathcal{L}t} \rho_0$ for the Lindbladian specified by Eq.~\eqref{eq:linblad}.  We choose the initial state $\rho_0 = \ketbra{0}{0}^{\otimes L}$ with $L=9$ qubits. The Lindbladian parameters are: $\lambda = 1$ and $\gamma_{XY} = 0.5$. The inset indicates that the initial state approaches the steady state exponentially fast.}
    \label{fig:plot_sym_xy}
\end{figure}
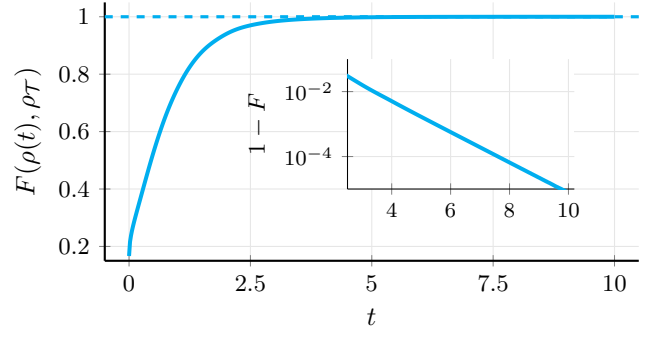

In contrast to the usual Lindbladian that involves a set of site-dependent local jump operators, the jump operator $L_{XY}$ here is a sum of local terms. This corresponds to the scenario where a single environment degree of freedom simultaneously couples to all physical sites, hence describing a collective dissipation process (See Fig.~\ref{fig:optics_state_spaces}a). This feature is essential for preserving the strong translation symmetry throughout the dynamics. Moreover, this particular jump operator prevents $\mathcal{L}$ from having additional strong symmetries. A simpler choice such as $\sum_i X_i$, for example, would preserve spatial reflection and spin flip symmetries as well. Therefore, the Lindbladian above describes a unital evolution with only strong translation symmetry, so that the MMIS $\rho_\tsym$ arises as the unique steady state in the zero momentum sector (see Ref.~\cite{zhang_stationary_2020, yoshida_uniqueness_2024, Sahu_2024_mmis} for discussion related to unital channels). This is confirmed based on the numerics shown in Fig. \ref{fig:plot_sym_xy}.

\textbf{Summary and discussion:} In this work, we have revealed a new mechanism for long-range entanglement of mixed states: whenever the rank of a mixed state $\rho$ with a strong symmetry $S$ is greater than the dimension of the space spanned by $S$-symmetric short-range entangled states, the mixed state $\rho$ must be LRE. This mechanism is intrinsic to mixed states without any pure-state counterpart, therefore providing a new perspective to explore the landscape of LRE mixed-states and likely leading to more examples beyond the ones we discuss in this work.

Within this new mechanism, we discuss the MMIS of translation $\rho_\tsym$ and several of its interesting features. In particular, it provides an example of SW-SSB of a non-onsite and non-gaugeable symmetry, whose diagnostic requires certain non-linear observables based on non-local charged operators. Exploring the stability of such SW-SSB and examples in higher dimensions is a natural future direction. 

As we have shown, the emergence of the MMIS $\rho_\tsym$ as a steady state requires collective dissipation in order to preserve strong translation symmetry. It would be interesting to explore more detailed properties of such Linbladians from other perspectives, such as the mixing time and the spectrum. Relatedly, the steady state of any Markovian dissipative dynamics always admits a natural holographic wavefunction in one higher spatial dimension as discussed in Ref.~\cite{lu_2025holographic}. Exploring the physical aspects of such a wavefunction may potentially provide new insights into such dynamics and the possible LSM-type constraints associated with it. 

Finally, as we have emphasized, the LRE mechanism based on dimensional constraints is a phenomenon specialized to mixed states. It is therefore natural to suspect that $\rho_\tsym$ is an intrinsically mixed quantum state, defined as one that is not two-way connected to any pure (either SRE or LRE) state through finite-depth local quantum channels. One such example is the one-dimensional CZX mixed state: $\rho_{\text{CZX}} \propto \one + \prod_i  X_i  \prod_i \text{CZ}_{i,i+1}$, due to its intriguing feature of being bipartite separable but tripartite long-range entangled \cite{lessa_2025_Multipartite}. Whether $\rho_\tsym$ and its generalizations are intrinsically mixed in this sense is an interesting question, which we leave for future work. 

\textbf{Note added.} While finalizing this manuscript, we became aware of a related work \cite{to_appear}, which will appear in the same arXiv posting.

\textbf{Acknowledgments}: We thank Chong Wang for helpful discussions. We also thank the programs ``Learning the Fine Structure of Quantum Dynamics in Programmable Quantum Matter'' and ``Noise-robust Phases of Quantum Matter'' at the Kavli Institute for Theoretical Physics (KITP), supported in part by NSF Grant No. PHY-2309135, as well as the workshop ``Bridging Classic and Contemporary Perspectives on Open Quantum Systems'' at the Simons Center for Geometry and Physics, where part of this work was completed. T.-C.L acknowledges the support of the RQS postdoctoral fellowship through the National Science Foundation (QLCI grant OMA-2120757). L.A.L. acknowledges support from the Natural Sciences and Engineering Research Council of Canada (NSERC) under Discovery Grant No. RGPIN-2020-04688 and No. RGPIN-2018-04380.  This work was also supported by an Ontario Early Researcher Award. Research at Perimeter Institute is supported in part by the Government of Canada through the Department of Innovation, Science and Industry Canada and by the Province of Ontario through the Ministry of Colleges and Universities.

\let\oldaddcontentsline\addcontentsline%
\renewcommand{\addcontentsline}[3]{}%
\bibliography{bibliography}
\let\addcontentsline\oldaddcontentsline%

\renewcommand\refname{Reference}

\clearpage
\widetext

\begin{center}
\textbf{\large Supplemental Material} \\~\\
 	\textit{   }
\end{center}

\tableofcontents

\section{Lemmas and proofs involved in Theorem \ref{thm:tsym_entangled}}

\subsection{The spans of homogeneous pure product states}\label{sec:span_product}

\begin{lemma}\label{lemma:homog_prod_span_psym}
    The set of homogeneous pure product states $S = \{\ket{\psi}^{\otimes L} \mid \ket{\psi} \in \Hilb \}$ spans the permutationally symmetric subspace: $\spn{S} = \psym \subseteq \tsym$.  
\end{lemma}

\begin{proof}
    Consider the uniform mixture of all homogeneous pure product states:
    \begin{align}
        M & \defeq \int \ketbra{\psi}{\psi}^{\otimes L} \ d\ket{\psi} \\
        & = \int U^{\otimes L} \ketbra{0\cdots 0}{0 \cdots 0} (U^\dagger)^{\otimes L} \ dU,
    \end{align}
    where the integral is over the Haar measure.  We aim to show that this operator is proportional to the projector onto the fully symmetric subspace: $M \propto P_\psym$.

    By Schur-Weyl duality between the on-site unitary group $U(\Hilb)$ and the site permutation group $S_L$, the total Hilbert space $\Hilb^{\otimes L}$ decomposes as
    \begin{equation}\label{eq:SW-decomp}
        \Hilb^{\otimes L} = \bigoplus_{\lambda} \sigma_{\lambda} \otimes \pi_\lambda, 
    \end{equation}
    where the direct sum is over all Young diagrams $\lambda$, which label the irrep $\sigma_\lambda$ of $U(\Hilb)$, and $\pi_\lambda$ of $S_n$. 

    By invariance under left multiplication of the Haar measure, $M$ commutes with all on-site unitaries $V^{\otimes L}$. Thus, by Schur's lemma, $M$ acts as a scalar multiple of the identity in each of the $\sigma_\lambda$ irreps:
    \begin{equation}
        M = \bigoplus_{\lambda} \one_{\sigma_{\lambda}} \otimes M_{\pi_{\lambda}}.
    \end{equation}
    Moreover, $M$ has a strong permutation symmetry. Hence, it is supported solely on the invariant irrep of $S_L$, corresponding to $\lambda = (L)$, while $M_{\pi_{\lambda}} = 0$ for $\lambda \neq (L)$. Since $\pi_{(L)}$ is one-dimensional, $M_{\pi_{(L)}}$ is proportional to the a complex number, thus proving $M \propto P_{\psym}$, which can only happen with $S$ spans $\psym$.    
\end{proof}

\noindent \textbf{Two remarks}: \\
(1) Ref.\cite{harrow2013church} provides a distinct proof for Lemma \ref{lemma:homog_prod_span_psym} without resorting to group theory and Schur-Weyl duality.\\
(2) Due to Lemma \ref{lemma:homog_prod_span_psym}, any pure state with permutation symmetry, including the long-range entangled ones, can be written as a linear combination of homogeneous product pure states. One simple example is the GHZ state $ \frac{1}{\sqrt{2}}(\ket{0}^{\otimes L} + \ket{1}^{\otimes L})$. A more interesting example is the W state $\ket{W} =\frac{1}{\sqrt{L}} \sum_{i=1}^L X_i \ket{0}^L  \propto \ket{1000...}+  \ket{0100...}+ \ket{0010...} + \cdots$. One can check that 

\begin{equation}
\ket{W} = \frac{1}{\sqrt{\mathcal{N}}}\sum_{k=0}^{L-1} \omega^{-k} \ket{\psi_k}^{\otimes L }, 
\end{equation}
where $\omega = e^{\frac{2\pi i}{L}}$ is $L$-th root of unity, $\ket{\psi_k} = \frac{1}{\sqrt{2}} ( \ket{0}+ \omega^k \ket{1})$, and $\mathcal{N}$ is a normalization constant. Therefore, W state is a superposition of the homogeneous product pure states $\ket{\psi_k}^{\otimes L }$.

\begin{proposition}\label{prop:psym_MMIS_fullsep}
    The permutation symmetry MMIS $\rho_\psym$ is fully separable.
\end{proposition}
\begin{proof}
    By the proof of Lemma \ref{lemma:homog_prod_span_psym}, we have that $\rho_{\psym} \propto \int \ketbra{\psi}{\psi}^{\otimes L} \ d\psi$. By Carathéodory's theorem, the integral can be replaced by a sum over a finite set of homogeneous product states (i.e., a quantum projective $L$-design): $\rho_{\psym} = \sum_\alpha p_\alpha \ketbra{\psi_\alpha}{\psi_\alpha}^{\otimes L}$, which gives a fully separable decomposition for $\rho_{\psym}$.
\end{proof}

\subsection{Proof of Lemma \ref{lemma:psym_dim}} \label{proof_lemma:psym_dim}

\begin{proof}
A basis for the permutationally symmetric subspace $\mathcal{H}_{\mathrm{sym}}$ is given by the symmetrized computational basis states
\begin{align}
\ket{\boldsymbol{\sigma}_{\mathrm{sym}}}
&\propto P_{\mathrm{sym}} \ket{\boldsymbol{\sigma}} \\
&\propto \sum_{\pi \in S_L} \ket{\sigma_{\pi(1)} \sigma_{\pi(2)} \cdots \sigma_{\pi(L)}}.
\end{align}

Each such state depends only on the occupation numbers $n_i$ of symbols $i \in \{0,\dots,d-1\}$, i.e., on the number of occurrences of each label in a length-$L$ string. In particular, distinct occupation-number configurations produce orthogonal and linearly independent states.

Therefore, a basis is labeled by integer tuples $(n_0,\dots,n_{d-1})$ with $n_i \ge 0$ and $\sum_{i=0}^{d-1}n_i = L$. By the stars-and-bars counting method, the number of such tuples is
\[
\binom{L+d-1}{d-1},
\]
which is therefore the dimension of $\mathcal{H}_{\mathrm{sym}}$.
\end{proof}

\subsection{Proof of Lemma \ref{lemma:tsym_dim}}\label{proof_lemma:tsym_dim}

\begin{proof}
Let $T$ denote the translation operator on $(\mathbb{C}^d)^{\otimes L}$. The projector onto the translationally symmetric subspace $\tsym$ is given by
\begin{align}
    P_{\tsym} = \frac{1}{L} \sum_{n=0}^{L-1} T^n.
\end{align}
Hence, its dimension is
\begin{align}
    \dim(\tsym) = \Tr(P_{\tsym}) = \frac{1}{L} \sum_{n=0}^{L-1} \Tr(T^n).
\end{align}

We now compute $\Tr(T^n)$. In the computational basis $\{ \ket{\boldsymbol{\sigma}} \}_{\boldsymbol{\sigma} \in \{0,\ldots,d-1\}^L}$, the operator $T^n$ acts by cyclically shifting the string $\boldsymbol{\sigma}$ by $n$ sites. Therefore, $\Tr(T^n)$ counts the number of basis states invariant under this shift:
\begin{align}
    \Tr(T^n) = \# \left\{ \boldsymbol{\sigma} \in \{0,\ldots,d-1\}^L : T^n \boldsymbol{\sigma} = \boldsymbol{\sigma} \right\}.
\end{align}

A string is invariant under a shift by $n$ sites if and only if it is periodic with period $\gcd(n,L)$. Thus, such strings are fully determined by their values on $\gcd(n,L)$ sites, giving
\begin{align}
    \Tr(T^n) = d^{\gcd(n,L)}.
\end{align}

Substituting back, we obtain
\begin{align}\label{eq:dim_T_appendix}
    \dim(\tsym ) = \frac{1}{L} \sum_{n=0}^{L-1} d^{\gcd(n,L)},
\end{align}
where $\gcd(n=0,L)=L$. 

We note that since $\gcd(n,L) \leq \frac{L}{2}$ for $n=1,2,..., L-1$, the dominant contribution in the summation comes from $n=0$. In particular, in the asymptotic limit  $L \to \infty$, one has 

\begin{equation}
 \dim(\tsym )  \to \frac{d^L}{L}.
\end{equation}

Finally, we note that this expression coincides with the number of distinct necklaces of length $L$ with $d$ colors, i.e., equivalence classes of strings under cyclic translations, as given by Pólya enumeration theorem.
\end{proof}

\subsection{Comparison between $\dim (\psym)$ and $\dim (\tsym{})$ }\label{appendix:table_dim}
For the qubit case $d=2$, and $n=1$, 2, or 3, we have $\psym = \tsym$. This can be checked by calculating their respective dimensions from Lemmas \ref{lemma:psym_dim} and \ref{lemma:tsym_dim} (See Table \ref{tab:pt-dimensions}). Hence, $\rho_\tsym = \rho_\psym$ is fully separable for those small system sizes. Otherwise, we have $\dim(\psym) < \dim(\tsym)$, which implies $\rho_{\tsym}$ is entangled for $n \geq 4$ from the argument above. For qudits of higher on-site dimension $d \geq 3$, $\psym = \tsym$ happens only for $n=1$ and 2, where the symmetry groups $S_n$ and $C_n$ themselves coincide.

\begin{table}[h!]
\centering
\begin{tabular}{c c @{~} c @{~} c}
\hline
$L$ & $\dim(\psym)$ & & $\dim(\tsym)$ \\
\hline
1 & 2 &  & 2 \\
2 & 3 &  & 3 \\
3 & 4 &  & 4 \\
4 & 5 & \makebox[0pt]{$<$} & 6 \\
5 & 6 & \makebox[0pt]{$<$} & 8 \\
6 & 7 & \makebox[0pt]{$<$} & 14 \\
\hline
\end{tabular}
\caption{Dimensions of $\psym$ and $\tsym$ for $L = 1$ to $6$ qubits ($d=2$).}
\label{tab:pt-dimensions}
\end{table}

\section{Proof of Lemma \ref{lemma:tsym_bipsep_implies_fullsep}: bipartite separabiltiy implies full separabiltiy }\label{append:bipartite_fully_separable} 

In this main text, we introduce Lemma \ref{lemma:tsym_bipsep_implies_fullsep}, which states that a translationally symmetric pure state that is bipartite separable with respect to contiguous bipartitioning $A|B$, i.e. $\ket{\Psi} = \ket{A}\ket{B}$, must be fully separable: $\ket{\Psi} = \ket{\psi}^{\otimes L}$. Here we will present the proof. For that, we first introduce a lemma that will be useful for our proof:  

\begin{lemma}[Overlapping pure marginals imply tripartite separability]\label{lemma:overlaping_pure_marginals}
    If $\rho_{ABC}$ is a state on a tripartite system $\Hilb_A \otimes \Hilb_B \otimes \Hilb_C$ with pure reduced density matrices on $AB$ and $BC$:
    \begin{align}
        \rho_{AB} & = \Tr_C[\rho_{ABC}] = \ketbra{\Psi_{AB}}{\Psi_{AB}}, \\
        \rho_{BC} & = \Tr_A[\rho_{ABC}] = \ketbra{\Phi_{BC}}{\Phi_{BC}},
    \end{align}
    then $\rho_{ABC}$ is a tripartite separable pure state $\rho_{ABC} = \ketbra{\Xi_{ABC}}{\Xi_{ABC}}$, $\ket{\Xi_{ABC}} = \ket{\psi_A}\ket{\xi_B}\ket{\phi_C}$.
\end{lemma}
\begin{proof}
Since $\rho_{AB} = \ketbra{\Psi_{AB}}{\Psi_{AB}}$ is pure, $\rho_{ABC} = \ketbra{\Psi_{AB}}{\Psi_{AB}} \otimes \rho_C$. In particular, $B$ and $C$ are uncorrelated, which implies that $\ket{\Phi_{BC}}$ factorizes into $\ket{\Phi_{BC}} = \ket{\xi_B}\otimes\ket{\phi_C}$. By the same reasoning applied to $\rho_{BC}$, we have $\ket{\Psi_{AB}} = \ket{\psi_A}\otimes \ket{\xi_B}$. The only possible extension compatible with both marginals is the tripartite separable pure state $\ket{\psi_A}\ket{\xi_B}\ket{\phi_C}$.
\end{proof}

Equipped with Lemma \ref{lemma:overlaping_pure_marginals}, we now present the proof of Lemma \ref{lemma:tsym_bipsep_implies_fullsep}: 

\begin{proof}
Without loss of generality, let us assume $A$ comes before $B$ in the site index order set by the translation operator $T$. Then, divide each region as $A = A_L \sqcup A'$ and $B = B_L \sqcup B'$, where $A_L$ and $B_L$ are the leftmost sites of regions $A$ and $B$, respectively (with $A'$ or $B'$ possibly empty).
    
By assumption, given the state $\ket{\Psi}$, the reduced density matrix on the region $A = A_L A'$ is also a pure state. By translation symmetry, the reduced density matrix on $A'B_L$ is pure as well because $A'B_L$ is the translation of $B$ by one site. Since the regions $A= A_L A'$ and $A'B_L$ overlaps on $A'$, by Lemma \ref{lemma:overlaping_pure_marginals}, the state on $A_L|A'|B_L$ is pure and tripartite separable. By a similar logic applied to regions $B$ and its translated version $B'A_L$, we conclude that the state is also tripartite separable with respect to the $B_L|B'|A_L$ tripartition. Both facts together imply four-partite separability w.r.t. $A_L|A'|B_L|B'$.  

    By repeating the procedure above, each time starting with a bipartitioning one site shifted to the right compared to the previous one (e.g., $A'B_L|B'A_L$ instead of $A|B$), we iteratively factorize $\ket{\Psi}$ into a product state, which has to be homogeneous due to translation symmetry. 
\end{proof}

\section{Proof for Lemma \ref{lemma:tsym_FDLU_bipsep_implies_fullsep}}\label{proof:tsym_FDLU_bipsep_implies_fullsep}

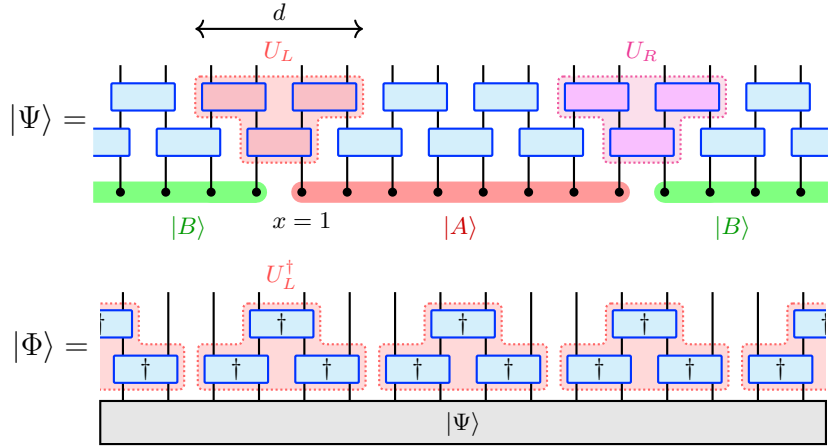
\begin{figure}[h]
    \centering
    \begin{tikzpicture}[
    x=1cm, y=1cm,
    box/.style={draw=blue!80!cyan, thick, fill=cyan!15, rounded corners=0.5pt},
    ulbox/.style={draw=blue!80!cyan, thick, fill=blue!10!red!25, rounded corners=0.5pt},
    urbox/.style={draw=blue!80!cyan, thick, fill=blue!10!magenta!25, rounded corners=0.5pt},
    dot/.style={fill=black, circle, inner sep=1.2pt},
    scale = 0.6
]

\draw[line width=8pt, green!50, line cap=butt] (0.4,-0.4) -- (4,-0.4);
\draw[line width=8pt, green!50, line cap=round] (3,-0.4) -- (4,-0.4);

\draw[line width=8pt, red!40, line cap=round] (5,-0.4) -- (12,-0.4);

\draw[line width=8pt, green!50, line cap=butt] (13,-0.4) -- (16.6,-0.4);
\draw[line width=8pt, green!50, line cap=round] (13,-0.4) -- (14,-0.4);

\node[green!60!black] at (2.5, -1.2) { $|B\rangle$};
\node[red!80!black] at (8.5, -1.2) { $|A\rangle$};
\node[green!60!black] at (14.5, -1.2) { $|B\rangle$};

\draw[thick, densely dotted, red!60, rounded corners=2pt, fill=red!15] 
    (2.65, 2.15) -- (6.35, 2.15) -- (6.35, 1.25) -- (5.35, 1.25) -- 
    (5.35, 0.25) -- (3.65, 0.25) -- (3.65, 1.25) -- (2.65, 1.25) -- cycle;

\draw[thick, densely dotted, magenta!80, rounded corners=2pt, fill=magenta!15] 
    (10.65, 2.15) -- (14.35, 2.15) -- (14.35, 1.25) -- (13.35, 1.25) -- 
    (13.35, 0.25) -- (11.65, 0.25) -- (11.65, 1.25) -- (10.65, 1.25) -- cycle;

\foreach \x in {1,...,16} {
    \draw[thick] (\x, -0.4) -- (\x, 2.4);
    \node[dot] at (\x, -0.4) {};
}

\node[below, inner sep = 0.8em] at (5, -0.4) {$x=1$};

\draw[box, fill=cyan!15] (0.4, 1.0) -- (1.2, 1.0) -- (1.2, 0.4) -- (0.4, 0.4); %
\draw[box] (1.8, 0.4) rectangle (3.2, 1.0); %
\draw[ulbox] (3.8, 0.4) rectangle (5.2, 1.0); %
\draw[box] (5.8, 0.4) rectangle (7.2, 1.0); %
\draw[box] (7.8, 0.4) rectangle (9.2, 1.0); %
\draw[box] (9.8, 0.4) rectangle (11.2, 1.0); %
\draw[urbox] (11.8, 0.4) rectangle (13.2, 1.0); %
\draw[box] (13.8, 0.4) rectangle (15.2, 1.0); %
\draw[box, fill=cyan!15] (16.6, 1.0) -- (15.8, 1.0) -- (15.8, 0.4) -- (16.6, 0.4); %

\draw[box] (0.8, 1.4) rectangle (2.2, 2.0); %
\draw[ulbox] (2.8, 1.4) rectangle (4.2, 2.0); %
\draw[ulbox] (4.8, 1.4) rectangle (6.2, 2.0); %
\draw[box] (6.8, 1.4) rectangle (8.2, 2.0); %
\draw[box] (8.8, 1.4) rectangle (10.2, 2.0); %
\draw[urbox] (10.8, 1.4) rectangle (12.2, 2.0); %
\draw[urbox] (12.8, 1.4) rectangle (14.2, 2.0); %
\draw[box] (14.8, 1.4) rectangle (16.2, 2.0); %

\node[red!70] at (4.5, 2.7) { $U_L$};
\node[magenta!90] at (12.5, 2.7) { $U_R$};
\draw[<->, thick] (2.65, 3.2) -- (6.35, 3.2) node[midway, above] { $d$};

\node at (-0.6, 1.2) {\large $|\Psi\rangle =$};

\end{tikzpicture}
    \begin{tikzpicture}[
    x=1cm, y=1cm,
    box/.style={draw=blue!80!cyan, thick, fill=cyan!15, rounded corners=0.5pt},
    scale=0.6
]

\node at (-0.6, 1.2) {\large $|\Phi\rangle =$};

\draw[thick, fill=gray!20] (0.5, -0.95) rectangle (16.5, 0);
\node at (8.5, -0.5) { $|\Psi\rangle$};

\begin{scope}
\clip (0.5, 0) rectangle (16.5, 3.2);

\foreach \xc in {0.5, 4.5, 8.5, 12.5, 16.5} {
    \draw[thick, densely dotted, red!60, rounded corners=2pt, fill=red!15] 
        (\xc - 0.85, 2.15) -- (\xc + 0.85, 2.15) -- 
        (\xc + 0.85, 1.25) -- (\xc + 1.85, 1.25) -- 
        (\xc + 1.85, 0.25) -- (\xc - 1.85, 0.25) -- 
        (\xc - 1.85, 1.25) -- (\xc - 0.85, 1.25) -- cycle;
}

\foreach \x in {1,...,16} {
    \draw[thick] (\x, 0) -- (\x, 2.4);
}

\foreach \k in {1,...,8} {
    \pgfmathsetmacro{\xleft}{2*\k - 1 - 0.2}
    \pgfmathsetmacro{\xright}{2*\k - 1 + 1.2}
    \draw[box] (\xleft, 0.4) rectangle (\xright, 1.0) node[midway] { $\dagger$};
}

\foreach \k in {0,...,4} {
    \pgfmathsetmacro{\xleft}{4*\k - 0.2}
    \pgfmathsetmacro{\xright}{4*\k + 1.2}
    \draw[box] (\xleft, 1.4) rectangle (\xright, 2.0) node[midway] { $\dagger$};
}
\end{scope}

\node[red!70] at (4.5, 2.8) { $U_L^\dagger$};

\end{tikzpicture}
    \caption{Schematic representation of the disentangling procedure. (Top) The state $\ket{\Psi}$ with regions $A$ and $B$. The unitary operation acting across the bipartition factorizes into $U_L \otimes U_R$, where $U_L$ and $U_R$ have support bounded by $d$ near the left ($x=1$) and right boundaries, respectively, up to unitaries that act on the bulk of $A$ and $B$. (Bottom) Construction of the short-range entangled ``bare'' state $\ket{\Phi}$ by systematically applying translated copies of the inverse boundary unitary, $U_L^\dagger$, to the original state $\ket{\Psi}$.}
    \label{fig:lemma_5}
\end{figure}

\begin{proof}
    Without loss of generality, we may assume $U$ only acts non-trivially near the two boundary points between $A$ and $B$, since we may redefine $\ket{A}$ and $\ket{B}$ to absorb the gates acting in the interior of each region. Specifically, $U$ factorizes into two small unitaries $U = U_L \otimes U_R$, each with support extending at most $d/2$ sites to either side of the left and right boundary points, respectively (See Fig. \ref{fig:lemma_5}). Let the site with index $x = 1$ be the one right next to the left boundary of $A$. By assumption, $|A|, |B| \gg d$.

    With $T$ being the translation operator by one site to the right, we define translated unitaries $U_{L}(d) \defeq (T^\dagger)^d U_{L} T^d$. Now, consider the ``bare'' state
    \begin{equation}\label{eq:naked_state}
        \ket{\Phi} \defeq \prod_{k = 0}^{\lfloor L / d \rfloor} U_L(k d)^\dagger 
        \ket{\Psi}.
    \end{equation}
    Let $R$ be a region around the left boundary of $A$ with size $|R| > 2d$. Since applying $U_L^\dagger$ to $\ket{\Psi}$ disentangles the left boundary, we have that $R$ is uncorrelated across $A|B$: $\rho_R = \Tr_{R^c} \ketbra{\Psi}{\Psi} = \rho_{R\cap B} \otimes \rho_{R \cap A}$. Due to the translation invariance of $\ket{\Psi}$ and the fact that we are also undoing the unitary $U_L(d)$, we also have that $R$ is uncorrelated across the translated bipartition $(A+d)|(B+d)$: $\rho_R = \rho_{R \cap (B + d)} \otimes \rho_{R \cap (A + d)}$. Hence, we have that $\rho_R = \rho_{B} \otimes \rho_{[1,d]} \otimes \rho_{R \cap (A+d)}$.

    Continuing this procedure by translating the system in steps of $d$ sites, we conclude that $\rho_{[1,kd]} = \rho_{[1,d]} \otimes \rho_{d+1, 2d} \otimes \cdots \otimes \rho_{[(k-1)d+1, kd]}$. Since the global state is pure, each of these mixed states are also pure, and we have that $\ket{\Phi}$ is a tensor product state:
    \begin{equation}
        \ket{\Phi} = \ket{\phi_{[1,d]}} \otimes \ket{\phi_{[d+1,2d]}} \otimes \cdots \otimes \ket{\phi_{[k^ * d, L]}}, 
    \end{equation}
    where $k^* = \lfloor L / d \rfloor$. By inverting the disentangling unitaries $U_L(k d)^\dagger$, we conclude that $\ket{\Psi} = \prod_{k = 0}^{\lfloor L / d \rfloor} U_L(k d) \ket{\Phi}$ is short-range entangled.

\end{proof}

\section{Correlation structure of $\rho_\tsym$}\label{appendix:correlation}

\subsection{Two-point function for prime $L$ }\label{appendix:two_point}
Given $\rho_{\tsym} = \frac{P_\tsym}{\mathcal{N}}$, with $P_\tsym = \sum_{k=0}^{L-1}  T^k$ and $\mathcal{N}$ the normalization constant, we label the lattice site from 0 to $L-1$. Below we will calculate the two-point correlation $\Tr  \rho_\tsym Z_0 Z_j^{-1}    =  \frac{1}{\mathcal{N}}\sum_{k=0}^{L-1}  \Tr  \left( T^k  Z_0 Z_j^{-1} \right)$. First, one has
\begin{equation} \label{eq:cyclic_trace_1}
\Tr ( T^kZ_0 Z_j^{-1} ) = \Tr (Z_{k} T^k Z_j^{-1}) =  \Tr (T^k  Z_{k}  Z^{-1}_j).  
\end{equation}

In other words, the two-point function between site $0$ and site $j$ is exactly equal to that between site $k$ and site $j$. Then one can repeat this process to find

\begin{equation}
\Tr ( T^k  Z_0 Z^{-1}_j) =  \Tr (T^k  Z_{mk}  Z^{-1}_j).   
\end{equation}
with $m$ being any positive integer and the subscript $mk$ in $Z_{mk}$ is defined as mod $L$. When $L$ and $k$ are co-prime, i.e. $\text{gcd}(L,k)=1$, any site $j$ can be obtained by $mk \text{ mod  } L$. Namely, the generator $T^k$ is ergodic: it can generate the entire cyclic group defined on $L$ sites. 

\begin{equation}\label{eq:cyclic_trace_2}
\Tr ( T^k Z_0 Z^{-1}_j) =  \Tr (T^k  Z_j  Z^{-1}_j )  = \Tr (T^k).   
\end{equation}

When $L$ is a prime number, the above is true for any $1\leq k \leq L-1$. Therefore,

\begin{equation}
\Tr (\rho_\tsym Z_0 Z^{-1}_j )   = \frac{1 }{\mathcal{N}}   \sum_{k=1}^{L-1}  \Tr ( T^k  )  = \frac{  \sum_{k=1}^{L-1}  \Tr ( T^k  )  }{  \sum_{k=0}^{L-1}  \Tr ( T^k  )  }  =   \frac{   \Tr ( T+ ... + T^{L-1}   )  }{  \Tr ( 1+ T+ ... + T^{L-1}   ) } 
\end{equation}
where we used $\Tr  Z_0Z_j^{-1}=0$  for $j\neq 0$. 

We consider the Pauli-Z basis $\ket{\sigma}\equiv \ket{ \{ \sigma_i \}  }$ with $\sigma_i \in \{ 0,1,2,...,d-1  \}$ to calculate the trace: $\Tr ( T^k ) = \sum_\sigma \bra{\sigma} T^k\ket{\sigma}$, which is equal to the number of $\sigma$ strings invariant under a shift over $k$ sites. The shift over $k$ sites can define an orbit, within which each of $\sigma$ strings can be connected using $k-$shift, and the number of disconnected orbits is $\text{gcd}(k,L)$. If the entire $\sigma$ string of size $L$ is invariant under $k-$shift, the only requirement is that within each orbit, the $\sigma_i$ should be the same for all $i$, which can take $d$ possible values.  Since there are $\text{gcd}(k,L)$ orbits, one finds $\Tr T^k  = 2^{\text{gcd}(k,L)} $ for $1\leq k \leq L-1$. As such, when $L$ is prime, $\Tr T^k  = d $, and one finds 

\begin{equation}\label{eq:two_site_correlation}
   \Tr (\rho_\tsym Z_0 Z_j^{-1})   = \frac{d (L-1)  }{  d^L  + d (L-1) }, 
\end{equation}
which is exponentially small in $L$.  

\subsection{$n$-point function for prime $L$ }
Here we present the calculation of $\Tr (\rho_{\tsym} \prod_{i=0}^{L-1}Z_i^{b_i})$ where each $b_i\in \{ 0,1,2,... , d-1\}$ and $\sum_{i=0}^{L-1} b_i=0 ~\text{mod } d$. Note that if $\sum_{i=0}^{L-1} b_i \neq 0 ~\text{mod } d$, the corresponding operator will have a vanishing expectation value due to the weak symmetry generated by $\prod_{i=0}^{L-1} X_i $, as discussed in the main text. We also note that when $b_i=0 ~\forall i$, the expectation value is simply $\Tr  \rho_\tsym =1$, so our discussion below will exclude this case as well.  

One central quantity that we need to calculate is $\Tr (T^k    \prod_{i=0}^{L-1}Z_i^{b_i})$. Again using the ergodicity of the $k$-shift generated by $T^k$, as discussed in the previous subsection (Eq.\ref{eq:cyclic_trace_1} to Eq.\ref{eq:cyclic_trace_2}), we find

\begin{equation}
\Tr (T^k    \prod_{i=0}^{L-1}Z_i^{b_i}) =  \Tr (T^k), 
\end{equation}
for $k=1,2,..., L-1$ when $L$ is prime. Following the same calculation as in the previous subsection (Appendix.\ref{appendix:two_point}), we find

\begin{equation}
   \Tr (\rho_\tsym   \prod_{i=0}^{L-1}Z_i^{b_i})   = \frac{d (L-1)  }{  d^L  + d (L-1) }, 
\end{equation}
when $\sum_{i=0}^{L-1} b_i=0 ~\text{mod } d$, and $\{b_i \}$ does not satisfy the condition $b_i=0 ~\forall i$.

\subsection{Exact results of single-site operator for arbitrary $L$}

Using similar techniques to the above, we will prove that, for arbitrary $L$ and $i \in \{1,2,\cdots,L\}$ an arbitrary site index, the one-site reduced density matrix of $\rho_{\tsym}$ is (exactly) maximally mixed 
\begin{equation}\label{eq:1site_reduced_tsym}
    \Tr_{\{i\}^C}[\rho_{\tsym}] = \frac{1}{d} \one_{\{i\}}.
\end{equation}

\textbf{Prime $L$}:

We will first start with the simpler case of prime $L$. In this case, for any local operator $O_i$ supported in site $i$, we have
\begin{equation}\label{eq:trace_translation_prime_n}
    \Tr[T^k O_i] = \begin{cases}
        \Tr_i[O_i], & \text{if } k\neq0 \\
        d^{L-1} \Tr_i[O_i], & \text{if } k=0
    \end{cases}
\end{equation}

The case $k=0$ is obvious. For the case $k\neq 0$, it reduces to a trace on a single-site Hilbert space, as illustrated in Fig. \ref{fig:trace_translation}.

\begin{figure}[h]
    \centering
    \includegraphics[width=0.5\linewidth]{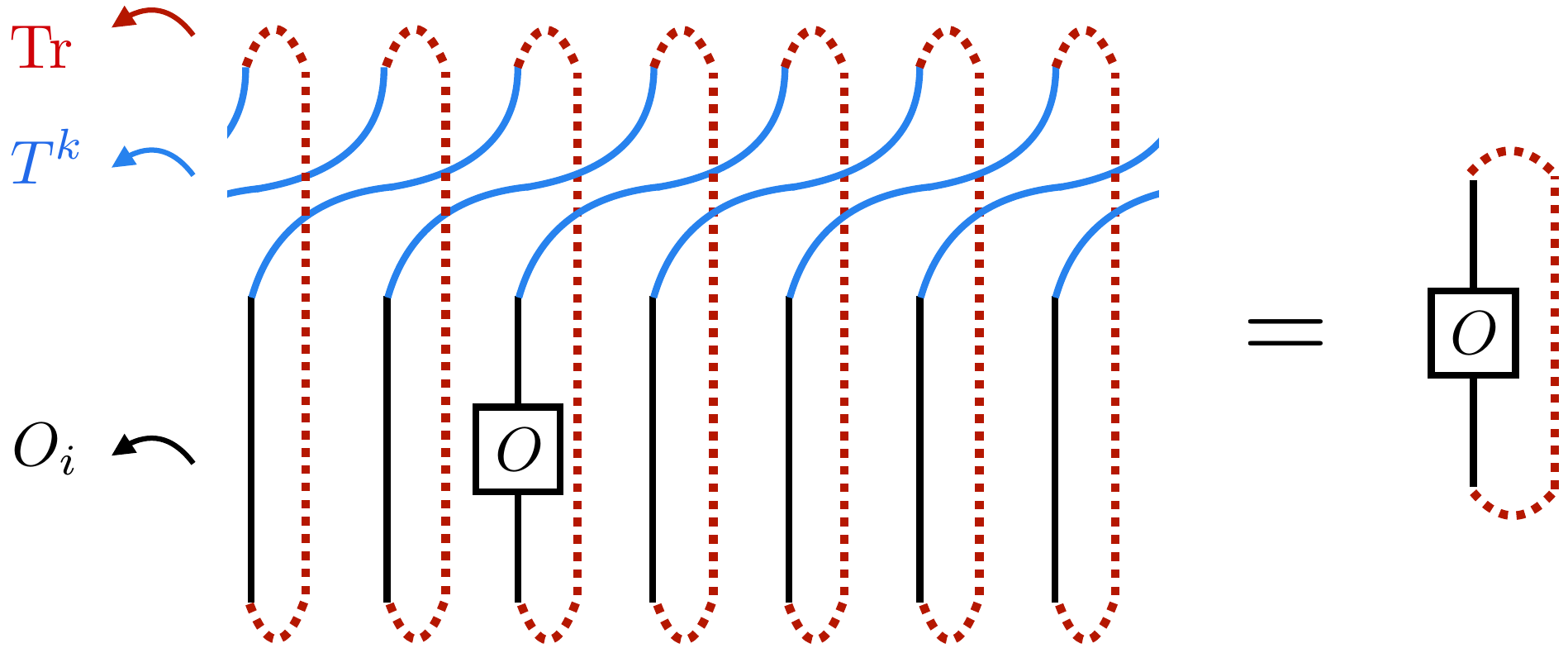}
    \caption{Visualization of $\Tr[T^k O_i] = \Tr_i [O_i]$ for $L=7$ sites, $i=3$, and $k = 2$.} 
    \label{fig:trace_translation}
\end{figure}

In algebraic terms,
\begin{equation}
    \Tr[T^k O_i] = \sum_{\sigma \in \{0,1,2,... d-1 \}^L} \braket{\sigma_{i+k}|O_i|\sigma_i} \prod_{j\neq i} \braket{\sigma_{j+k} | \sigma_j}.
\end{equation}
The product in the RHS above requires that $\forall j \neq i, \sigma_j = \sigma_{j+k}$. Since $k \neq 0 \pmod{L}$, then $i+k \neq i \pmod{L}$, which means $\sigma_{i+k} = \sigma_{i+2k} = \cdots = \sigma_{i+mk}$, as long as $i+mk \neq i \pmod{L}$, but this is true for all $1 \leq m \leq L-1$, due to the primality of $L$. Since $\sigma_{i+(L-1)k} = \sigma_{i-k} = \sigma_i$ as well (because $i-k \neq i \pmod{L}$), then all $\sigma_j$ are equal to a single $\sigma\in \{0,1,2,... d-1 \}$:
\begin{equation}
    \Tr[T^k O_i] = \sum_{\sigma \in \{0,1,2,... d-1 \}} \braket{\sigma|O_i|\sigma} = \Tr_i[O_i].
\end{equation}

Hence, 
\begin{equation}
    \Tr[\rho_{\tsym} O_i] = \frac{\sum_{k=0}^{n-1} \Tr[T^k O_i]}{\sum_{k=0}^{n-1} \Tr[T^k]} = \Tr_i[O_i] \frac{d^{n-1} + (n-1)}{d^n + d(n-1)} = \frac{1}{d} \Tr_i[O_i],
\end{equation}
which proves Eq. \eqref{eq:1site_reduced_tsym} for prime $L$.

\textbf{Arbitrary $L$}:

For arbitrary $L$, we can generalize Eq. \eqref{eq:trace_translation_prime_n} to
\begin{equation}
    \Tr[T^k O_i] = \Tr_i[O_i] d^{\gcd(k,L) -1} = \frac{1}{d} \Tr_i[O_i] \Tr[T^k],
\end{equation}
for arbitrary $k \in \Z$. Then, it immediately follows that
\begin{equation}
    \Tr[\rho_{\tsym} O_i] = \frac{\sum_{k=0}^{n-1} \Tr[T^kO_i]}{\sum_{k=0}^{n-1} \Tr[T^k]} = \frac{1}{d} \Tr_i[O_i].
\end{equation}

\subsection{Proof for Theorem \ref{thm:rdm}: reduced density matrix is maximally mixed}
Given $\rho_\tsym$, in this subsection, we derive an upper bound for the trace difference between the reduced density matrix and the maximally mixed state on a subregion $A$. We will first consider the case for prime $L$ as a warm-up, and then discuss the case with arbitrary $L$, which requires a more involved calculation.

\subsubsection{For prime $L$}

Let $O_A$ be an arbitrary operator in region $A$. Then, by the same mechanism of Figure \ref{fig:trace_translation}, for all $k \neq 0 \pmod{L}$ there is a permutation $\pi_k \in S_{A}$ of the indices of $A$ such that
\begin{align}
    \Tr[T^k O_A] =   \Tr_A[P_{\pi_k} O_A] =    \sum_{i_1,\cdots,i_{|A|}} [O_A]_{i_1,\ldots,i_{|A|}}^{i_{\pi(1)},\ldots,i_{\pi(|A|)}}     
\end{align}
where $P_{\pi_k} \in \mathcal{L}(\Hilb_A)$ is the permutation matrix associated with $\pi_k$, and in the right-hand side of the above equation, we have neglected the $k$ dependence for $\pi_k$ for notational simplicity. The above equation can also be visualized as shown in Fig.\ref{fig:prime_L_diagram}.

\begin{figure}[h]
    \centering
    \includegraphics[width=0.6\linewidth]{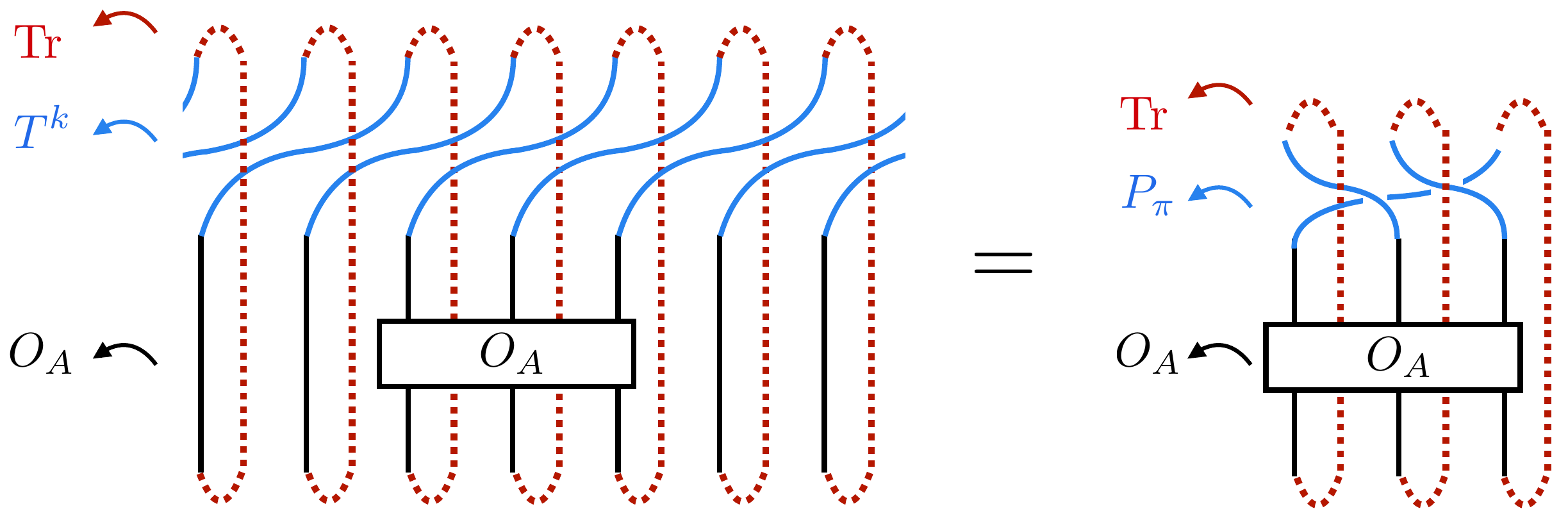}
    \caption{Visualization of $ \Tr[T^k O_A] =   \Tr_A[P_{\pi} O_A]$ for $L=7$ sites, 3-site operator $O_A$, and $k = 2$.} 
    \label{fig:prime_L_diagram}
\end{figure}

Moreover, using $|\mathrm{Tr}(X)| \le \|X\|_1$ and the unitary invariance of the trace norm, one has $| \Tr_A[P_{\pi} O_A]  | \leq \norm{P_{\pi} O_A}_1 = \norm{O_A}_1$. Then,

\begin{equation}
    \Tr[\rho_{\tsym} O_A] = \frac{\Tr_A[O_A]d^{L-|A|} +\sum_{k=1}^{L-1} \Tr_A[P_{\pi_k}O_A]}{d^L + d(L-1)} \leq  \frac{   d^{-|A|}\Tr_A [O_A]  }{  1+   d^{-L+1 } (L-1)   } +  \frac{d^{-L} (L-1) \norm{O_A}_1 }{ 1+   d^{-L+1 } (L-1)    }. 
\end{equation}

It follows that 
\begin{equation}
|\Tr[\rho_{\tsym} O_A]  -   \frac{1}{d^{|A|}}\Tr_A[O_A]  |  \leq \norm{O_A}_1 O(L d^{-L- |A|} ) +   \norm{O_A}_1 O(L d^{-L} ).
\end{equation}

Using $\norm{O_A}_1 \leq \norm{O_A}_\infty d^{|A|}$, one finds

\begin{equation}
|\Tr[\rho_{\tsym} O_A]  -   \frac{1}{d^{|A|}}\Tr_A[O_A]    |  \leq \norm{O_A}_\infty  O(L d^{-L} ) +   \norm{O_A}_\infty O(L d^{-L +|A|} ).
\end{equation}

Based on the definition of the trace norm: $\norm{X}_1  = \text{sup}_{\norm{O}_{\infty}<1} |\Tr (  XO)|$, the above bound for the difference of expectation values leads to the bound for the trace distance between the two density matrices:

\begin{equation}
\norm{\Tr_{A^C}[\rho_\tsym] - \frac{\one}{d^{|A|}}}_1   \leq   O(L d^{-L} ) +  O(L d^{-(L -|A|)} )  =  O(L d^{-(L -|A|)} ) .
\end{equation}

Therefore, as long as $\frac{|A|}{L}$ is a finite fraction less than $1$, in the thermodynamic limit, the reduced density matrix $\Tr_{A^C}[\rho_\tsym] $ is a maximally mixed state.

\subsubsection{For arbitrary $L$}

We now generalize the proof for arbitrary $L$. Compared to the prime case, the main new feature is the appearance of disconnected cycles in the permutation induced by $T^k$. 

Our goal is to upper bound $\Tr[T^k O_A]$ when $k \neq 0 \pmod{L}$. If $L$ is prime, as we saw before, it is equivalent to a reduced trace over a permuted operator: $\Tr[T^k O_A] = \Tr_A  [ P_{\pi_k} O_A]$. For non-prime $L$, there may be cycles of strings that do not touch any incoming or outgoing leg of $O_A$; see Fig. \ref{fig:non_prime_diagram}. 

\begin{figure}[h]
    \centering
    \includegraphics[width=0.97\linewidth]{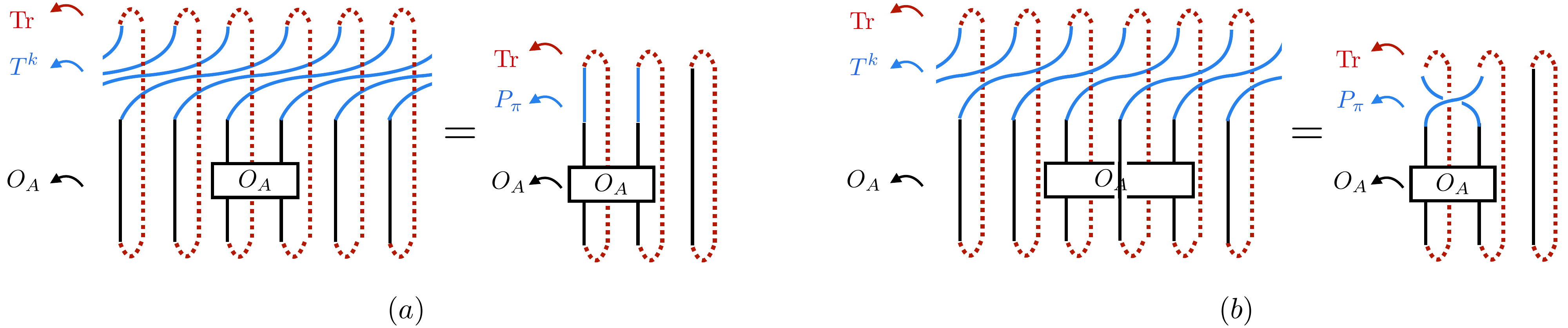}
    \caption{Visualization of $ \Tr[T^k O_A] =   C_{L,k,A} \Tr_A[P_{\pi_k} O_A]$ for non-prime $L$. Note that we have neglected the subscript $k$ when expressing $P_{\pi_k}$ in the diagram. (a) $L=6$, $k=4$, and the two-body operator $O_A$ acts on two neighboring sites. $C_{L,k,A} = d$ since there is a closed string that does not touch any incoming or outgoing leg of $O_A$. The permutation matrix $P_\pi$ is an identity. (b) $L=6$, $k=2$, and the two-body operator $O_A$ acts on the 3rd and 5th sites. $C_{L,k,A} = d$ since there is a closed string that does not touch any incoming or outgoing leg of $O_A$. The permutation matrix $P_\pi$ generates a swap between the two sites.} 
    \label{fig:non_prime_diagram}
\end{figure}

These cycles will contribute to an additional prefactor $C_{L,k,A}$: 
\begin{equation}
    \Tr[T^k O_A] = C_{L,k,A} \Tr_A[ P_{\pi_k}  O_A].
\end{equation}

To evaluate $C_{L,k,A}$,  we analyze the corresponding string diagram (see e.g., Fig.\ref{fig:trace_translation}). Each cycle contributes a factor of $d$ (the trace of identity, or a loop in the string picture) to $C_{L,k,A}$. Moreover, each cycle is in one-to-one correspondence with an arithmetic progression \((a_m) \subseteq \mathbb{Z}_L\) of the form \(a_m = a_0 + mk\), with common difference \(k\), such that all \(a_m \in A^C\), where $A^{C}$ denotes the complement of the subregion $A$. Let $N^C_{L,k,A}$ be the number of such cycles, one has 

\begin{equation}
  C_{L,k,A} = d^{N^C_{L,k,A}}.  
\end{equation} 
Since each cycle necessarily has length $L/\gcd(L,k)$, the number of such cycles that can fit $A^C$ is bounded by 
\begin{equation}\label{eq:number_cycles}
    N^C_{L,k,A} \leq  \frac{|A^C|}{L/\gcd(L,k)} = \gcd(L,k) \left(1- \frac{|A|}{L} \right).
\end{equation}

Using the definition of $\rho_\tsym$, one has 
\begin{align}
    \Tr[\rho_{\tsym} O_A] & =  \frac{\Tr_A[O_A]d^{L-|A|} + \sum_{k=1}^{L-1} C_{L,k,A} \Tr_A[ P_{\pi_k}O_A]}{d^L + \sum_{k=1}^{L-1} \Tr[T^k]} \\
    & = \frac{1}{d^A} \Tr_A[O_A] \frac{1}{1+\sum_{k=1}^{L-1}d^{\gcd(L,k)-L}} + \frac{d^{-L}\sum_{k=1}^{L-1} C_{L,k,A} \Tr_A[P_{\pi_k}O_A]}{1+\sum_{k=1}^{L-1}d^{\gcd(L,k)-L}}.
\end{align}
For $1\leq k \leq L-1$, we have $\gcd(L,k) \leq  \frac{L}{2}$ \footnote{More generally, if $p$ is the smallest prime factor of $L$, then $\gcd(L,k) \leq L/p$ for $1 \leq k \leq L-1$.}, so  that 

\begin{equation}
\sum_{k=1}^{L-1}d^{\gcd(L,k)-L} \leq (L-1)d^{-L/2} = O(Ld^{-L/2}).
\end{equation} 

On the other hand, Eq. \eqref{eq:number_cycles} implies 
\begin{equation}
d^{-L}\sum_{k=1}^{L-1} C_{L,k,A} \leq O(L d^{-(L+ |A|)/2}).
\end{equation}

Using the above two equations, one finds

\begin{align}
    |\Tr[\rho_\tsym O_A] - \frac{1}{d^{|A|}} \Tr_A[O_A]|& \leq \norm{O_A}_1  O(Ld^{-L/2  -|A|  }) + \norm{O_A}_1 O(Ld^{-(L+|A|)/2}),
\end{align}
where we have also used $| \Tr_A[P_{\pi} O_A]  | \leq \norm{P_{\pi} O_A}_1 = \norm{O_A}_1$. Moreover, with $\norm{O_A}_1 \leq \norm{O_A}_\infty d^A$, one finds

\begin{align}
    |\Tr[\rho_\tsym O_A] - \frac{1}{d^{|A|}} \Tr_A[O_A]|& \leq \norm{O_A}_\infty  O(Ld^{-L/2   }) + \norm{O_A}_\infty O(Ld^{-(L-|A|)/2})
\end{align}

Based on the definition of trace norm: $\norm{X}_1  = \text{sup}_{\norm{O}_{\infty}<1} |\Tr (  XO)|$, the above bound for the difference of expectation values leads to the bound for the trace distance between the two density matrices:

\begin{equation}
\norm{\Tr_{A^C}[\rho_\tsym] - \frac{\one}{d^{|A|}}}_1  \leq O(Ld^{-L/2 }) +  O(Ld^{-(L-|A|)/2}) =   O(Ld^{-(L-|A|)/2}).  
\end{equation}
Therefore, as long as $\frac{|A|}{L}$ is a finite fraction less than $1$, in the thermodynamic limit, the reduced density matrix $\Tr_{A^C}[\rho_\tsym] $ is a maximally mixed state. 

\section{Operator space entanglement}\label{appendix:operator_space_entanglement}

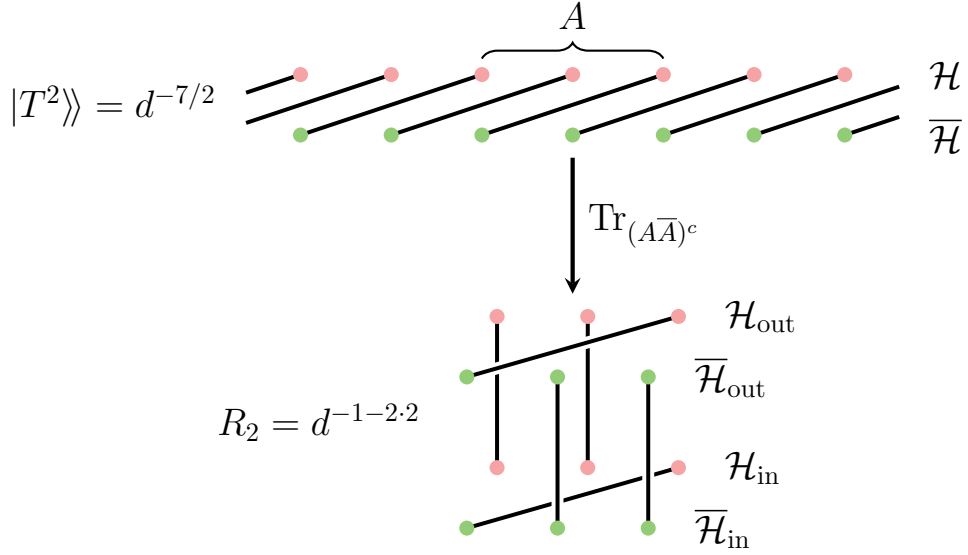
\begin{figure}[t]
    \centering
    \begin{tikzpicture}

\def\topL{7}
\def\topShift{0}
\def\topX{0}
\def\topY{5.2}

\drawblock{top}{\topL}{\topShift}{\topX}{\topY}

\foreach \i in {1,...,5} {
    \pgfmathtruncatemacro{\nexti}{\i+2}
    \connectIntra{top}{\i}{\nexti}
}

\path [name path=leftPlane] (\topX + 0.4*\xdistance, \topY - 2) -- (\topX + 0.4*\xdistance, \topY + 2);
\path [name path=rightPlane] (\topX + \topL*\xdistance + 0.6*\xdistance, \topY - 2) -- (\topX + \topL*\xdistance + 0.6*\xdistance, \topY + 2);

\path [name path=lineT1] (top-T-1) -- ++(-2*\xdistance - \topShift, -\ydistance);
\path [name path=lineT2] (top-T-2) -- ++(-2*\xdistance - \topShift, -\ydistance);
\draw[thickline, name intersections={of=leftPlane and lineT1, by=intT1}] (intT1) -- (top-T-1);
\draw[thickline, name intersections={of=leftPlane and lineT2, by=intT2}] (intT2) -- (top-T-2);

\pgfmathtruncatemacro{\topLminusOne}{\topL-1}
\path [name path=lineB6] (top-B-\topLminusOne) -- ++(2*\xdistance + \topShift, \ydistance);
\path [name path=lineB7] (top-B-\topL) -- ++(2*\xdistance + \topShift, \ydistance);
\draw[thickline, name intersections={of=rightPlane and lineB6, by=intB6}] (top-B-\topLminusOne) -- (intB6);
\draw[thickline, name intersections={of=rightPlane and lineB7, by=intB7}] (top-B-\topL) -- (intB7);

\node[left] at ([xshift=-1.0cm]top-B-1 |- {0, \topY+\ydistance/2}) {\Large $|T^2\rangle\!\rangle = d^{-7/2}$};
\node[right] at ([xshift=1.0cm]top-T-\topL) {\Large $\mathcal{H}$};
\node[right] at ([xshift=1.0cm]top-B-\topL) {\Large $\overline{\mathcal{H}}$};

\draw [decorate,decoration={brace,amplitude=6pt}, thick] 
    ([yshift=6pt]top-T-3.north west) -- ([yshift=6pt]top-T-5.north east) 
    node [midway,yshift=16pt, font=\Large] {$A$};

\coordinate (arrowStart) at ([yshift=-0.3cm]top-B-4);
\coordinate (arrowEnd) at ([yshift=-1.8cm]arrowStart);
\draw[->, ultra thick, >=stealth] (arrowStart) -- (arrowEnd) node[midway, right, xshift=2pt] {\Large $\text{Tr}_{(A \overline {A})^c}$};

\def\botL{3}
\def\botShift{0.4}
\def\outX{2.2}
\def\outY{2.0}
\def\inX{2.2}
\def\inY{0}

\drawblock{out}{\botL}{\botShift}{\outX}{\outY}
\drawblock{in}{\botL}{\botShift}{\inX}{\inY}

\connectIntra{in}{1}{3}

\connectInter[cross]{out}{T}{1}{in}{T}{1}
\connectInter[cross]{out}{T}{2}{in}{T}{2}
\connectInter[cross]{out}{B}{2}{in}{B}{2}
\connectInter[cross]{out}{B}{3}{in}{B}{3}

\connectIntra[cross]{out}{1}{3}

\coordinate (botCenterY) at (0, \outY/2 + \inY/2 + \ydistance/2);
\node[left] at ([xshift=-0.5cm]out-B-1 |- botCenterY) {\Large $R_2 = d^{-1-2\cdot 2}$};

\node[right] at ([xshift=0.5cm]out-T-\botL) {\Large $\mathcal{H}_{\text{out}}$};
\node[right] at ([xshift=0.5cm]out-B-\botL) {\Large $\overline{\mathcal{H}}_{\text{out}}$};
\node[right] at ([xshift=0.5cm]in-T-\botL) {\Large $\mathcal{H}_{\text{in}}$};
\node[right] at ([xshift=0.5cm]in-B-\botL) {\Large $\overline{\mathcal{H}}_{\text{in}}$};

\end{tikzpicture}
    \caption{Double Hilbert space representation of the translation operator $T^n$ (top) and its reduced density matrix $R_n$ on a contiguous subsystem $A \overline{A}$ (bottom). For a system of size $L=7$ and translation index $n=2$, the normalized pure state $\kett{T^2} \in \Hilb \otimes \overline{\Hilb}$ is depicted by slanted lines connecting the ``bra'' space $\overline{\Hilb}$ to the ``ket'' space $\mathcal{H}$, with a normalization prefactor of $d^{-7/2} = d^{-L/2}$. Tracing out the complement region $(A \overline{A})^c$ contracts its 'out' and 'in' spaces, generating the reduced operator $R_2 : \Hilb_{\text{in}} \otimes \overline{\Hilb}_{\text{in}} \to \Hilb_{\text{out}} \otimes \overline{\Hilb}_{\text{out}}$. The vertical lines represent identity operators that emerge from the Bell pairs straddling the boundary of region $A$. The final normalization $d^{-1-2\cdot 2} = d^{-(|A|-n) - 2n}$ ensures $\Tr(R_2) = 1$ by assigning a factor of $d^{-1}$ to each identity operator and each Bell pair.}
    \label{fig:translation_double_space}
\end{figure}

Here, we calculate the operator space entanglement of $\rho_\tau$ using both von Neumann and Rényi entropies. First, we introduce some basic terminology: the doubled state associated to $\rho_\tsym$ is $\kett{\rho_\tsym} \propto \sum_{n=0}^{L-1} \kett{T^n} \in \Hilb \otimes \overline{\Hilb}$, where $\kett{T^n}$ is the state associated to the translation operator, which can be expressed as a series of Bell states connecting sites related by translation: $\kett{T^n} = \bigotimes_{i=1}^L \kett{(i+n,\overline{i})}$. We assume that the subregion $A$ is an interval and, without loss of generality, that $|A| \leq L/2$. The reduced density matrix in $A$ is $R_A = \Tr_{(A\overline{A})^c} \kettbbra{\rho_\tsym}{\rho_\tsym} \propto  \sum_{m,n=0}^{L-1} R_{m,n}$, where $R_{m,n} = \Tr_{(A\overline{A})^c}  \kettbbra{T^m}{T^n}$. Here, $\overline{A}$ is the corresponding region $A$ in the bra space in $\overline{\Hilb}$. For simplicity, we denote the diagonal terms $R_{n,n}$ by $R_n \defeq R_{n,n}$. See Fig.~\ref{fig:translation_double_space} for a visual illustration of $\kett{T^n}$ and $R_n$ in terms of string diagrams, which will be useful for future calculations.

The rest of the section will proceed in three steps:
\begin{enumerate}
    \item Prove that the off-diagonal terms $R_{m,n}$ make up a vanishingly small contribution to $R_A$;
    \item Prove that the diagonal terms $R_n \equiv R_{n,n}$ are approximately orthogonal;
    \item Finally, calculate the entropy of $R_A$ by employing the approximations above.
\end{enumerate}
For the proofs below, we will make heavy use of the fact that the operators involved can be viewed as string diagrams in the tensor network sense, since they are formed by Bell pairs and identity matrices.

\begin{figure}
    \centering
    \def\blocksep{2.0}
\begin{tikzpicture}
\def\shift{0.3}

\def\X{2} %
\def\L{7}
\def\topY{0} %
\def\botY{-\blocksep} %

\node[left] at (\X + 0.2, -\blocksep/2 + \ydistance/2) {\Large $R_{2,1} = d^{-7}$};

\drawblock{topleft}{\L}{\shift}{\X}{\topY}
\drawblock{botleft}{\L}{\shift}{\X}{\botY}

\begin{pgfonlayer}{background}
\foreach \i in {1,2,6,7} {
    \draw[reddotted] (topleft-T-\i) -- (botleft-T-\i);
    \draw[reddotted] (topleft-B-\i) -- (botleft-B-\i);
}
\end{pgfonlayer}

\draw [decorate,decoration={brace,amplitude=6pt}, thick] 
    ([yshift=6pt]topleft-T-3.north west) -- ([yshift=6pt]topleft-T-5.north east) 
    node [midway,yshift=16pt, font=\Large] {$A$};

\connectIntra{topleft}{1}{3}
\connectIntra{topleft}{2}{4}
\connectIntra{topleft}{3}{5}
\connectIntra{topleft}{4}{6}
\connectIntra{topleft}{5}{7}

\connectIntra{botleft}{1}{2}
\connectIntra{botleft}{2}{3}
\connectIntra{botleft}{3}{4}
\connectIntra{botleft}{4}{5}
\connectIntra{botleft}{5}{6}
\connectIntra{botleft}{6}{7}

\path [name path=leftPlane] (\X + 0.4*\xdistance, \botY - 0.5) -- (\X + 0.4*\xdistance, \topY + \ydistance + 0.5);
\path [name path=rightPlane] (\X + \L*\xdistance + 0.6*\xdistance, \botY - 0.5) -- (\X + \L*\xdistance + 0.6*\xdistance, \topY + \ydistance + 0.5);

\path [name path=toplineT1] (topleft-T-1) -- ++(-2*\xdistance - \shift, -\ydistance);
\path [name path=toplineT2] (topleft-T-2) -- ++(-2*\xdistance - \shift, -\ydistance);
\draw[thickline, name intersections={of=leftPlane and toplineT1, by=topintT1}] (topintT1) -- (topleft-T-1);
\draw[thickline, name intersections={of=leftPlane and toplineT2, by=topintT2}] (topintT2) -- (topleft-T-2);

\pgfmathtruncatemacro{\LminusOne}{\L-1}
\path [name path=toplineB6] (topleft-B-\LminusOne) -- ++(2*\xdistance + \shift, \ydistance);
\path [name path=toplineB7] (topleft-B-\L) -- ++(2*\xdistance + \shift, \ydistance);
\draw[thickline, name intersections={of=rightPlane and toplineB6, by=topintB6}] (topleft-B-\LminusOne) -- (topintB6);
\draw[thickline, name intersections={of=rightPlane and toplineB7, by=topintB7}] (topleft-B-\L) -- (topintB7);

\path [name path=botlineT1] (botleft-T-1) -- ++(-\xdistance - \shift, -\ydistance);
\draw[thickline, name intersections={of=leftPlane and botlineT1, by=botintT1}] (botintT1) -- (botleft-T-1);

\path [name path=botlineB7] (botleft-B-\L) -- ++(\xdistance + \shift, \ydistance);
\draw[thickline, name intersections={of=rightPlane and botlineB7, by=botintB7}] (botleft-B-\L) -- (botintB7);

\node at (\X + \L*\xdistance + 0.6*\xdistance + 0.8, -\blocksep/2 + \ydistance/2) {\Large $= d^{-7}$};

\def\rightL{3}
\def\rightX{\X + \L*\xdistance + 0.6*\xdistance + 0.5}

\drawblock{topright}{\rightL}{\shift}{\rightX}{\topY}
\drawblock{botright}{\rightL}{\shift}{\rightX}{\botY}

\begin{pgfonlayer}{background}
    \connectIntra{topright}{3}{1}
    \connectInter{topright}{T}{2}{botright}{T}{1}
    \connectIntra{botright}{2}{3}
    \connectIntra{botright}{1}{2}
\end{pgfonlayer}

\connectIntra[cross]{topright}{1}{3}

\connectInter[cross]{topright}{B}{2}{botright}{B}{3}

\end{tikzpicture}
    \caption{Accompanying illustration for the proof of Lemma \ref{lemma:off_diagonal_vanish}. Depicted here is the double Hilbert space operator $R_{2,1} = \Tr_{(A\overline{A})^c}  \kettbbra{T^2}{T}$ for $L=7$. It has $T_{2,1} = 2$ horizontal lines in the outgoing space (on the top), $B_{2,1} = 2$ horizontal lines in the incoming space (on the bottom), and $V_{2,1} = 2$ vertical lines. No loops are formed in this contraction, so $L_{2,1} = 0$. One can further check that $\norm{R_{2,1}} =  d^{-3}$.}
    \label{fig:off_diagonal}
\end{figure}
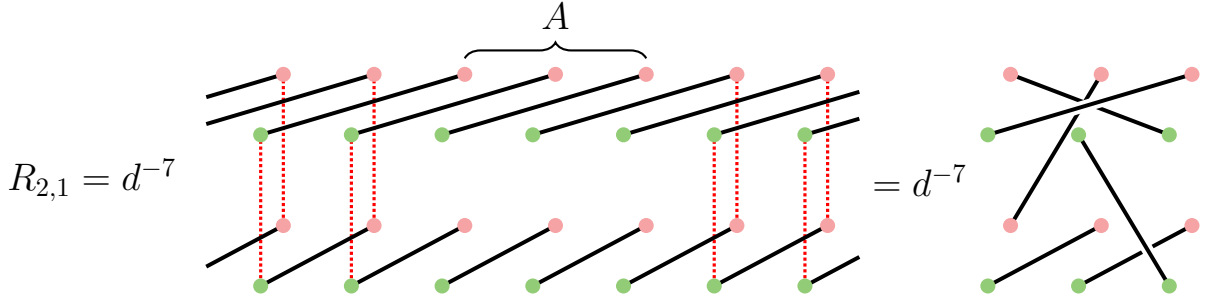

For the first point, we start with the following lemma
\begin{lemma}\label{lemma:off_diagonal_vanish}
$\norm{R_{m,n}}_1 \leq d^{- |A| / 2}$ for $m \neq n$.
\end{lemma}
\begin{proof}
    First, we refer the reader to Fig.~\ref{fig:off_diagonal}, which exemplifies much of the steps below.

    After tracing out $(A\overline{A})^c$, $R_{m,n}$ is a string diagram of normalization factor $d^{-L + L_{m,n}}$, with $L_{m,n}$ indicating the number of loops formed in the contraction process. Furthermore, it consists of $T_{m,n}$ horizontal lines in the outgoing space (on the top), $B_{m,n}$ horizontal lines in the incoming space (on the bottom), and $V_{m,n}$ vertical lines, which may connect points diagonally.
    
    To estimate the trace norm $\norm{R_{m,n}}_1 = \Tr \sqrt{R_{m,n}^\dagger R_{m,n}}$, we now examine the resulting string diagram formed by $R_{m,n}^\dagger R_{m,n}$. Since $R_{m,n}^\dagger$ is the same as $R_{m,n}$, but with incoming and outgoing degrees of freedom flipped, we have that the string diagram of $R_{m,n}^\dagger R_{m,n}$ contains $B_{m,n}$ horizontal lines at the bottom and, in parallel, $B_{m,n}$ at the top, and $V_{m,n}$ straight vertical lines. In the contraction of the middle layer, $R_{m,n}^\dagger R_{m,n}$ acquires a further factor of $d^{T_{m,n}}$. The final normalization factor is, then, $\mathcal{N} = d^{-2L + 2 L_{m,n} + T_{m,n}}$.

    From the form of the string diagram of $R_{m,n}^\dagger R_{m,n}$, we have that $(R_{m,n}^\dagger R_{m,n})^2 = \mathcal{N} d^{B_{m,n}} R_{m,n}^\dagger R_{m,n}$. By taking the square root of this equation, we have 
    \begin{equation}
        \sqrt{R_{m,n}^\dagger R_{m,n}} = \mathcal{N}^{-1/2} d^{-B_{m,n}/2} R_{m,n}^\dagger R_{m,n}.
    \end{equation}
    Furthermore, we have that 
    \begin{equation}
        \Tr[R_{m,n}^\dagger R_{m,n}] = \mathcal{N} d^{B_{m,n} + V_{m,n}}.
    \end{equation}
    Hence,
    \begin{align}
        \norm{R_{m,n}}_1  & = \Tr \sqrt{R_{m,n}^\dagger R_{m,n}} \\ 
        & = \mathcal{N}^{1/2} d^{\frac{B_{m,n}}{2} + V_{m,n}} \\
        & = d^{-L + L_{m,n} + \frac{T_{m,n}}{2} + V_{m,n} + \frac{B_{m,n}}{2}} \\
        & = d^{-(L-|A|) + L_{m,n} + \frac{V_{m,n}}{2}},
    \end{align}
    where we have used that $T_{m,n} + V_{m,n}/2 = B_{m,n} + V_{m,n}/2 =  |A|$ in the last line. 

    Now, we upper bound the one-norm as a function of $L$ and $|A|$. For $V_{m,n}$, it is simply bounded as $V_{m,n} \leq \min(|m|,|n|,|A|)$, where here the range of $m$ and $n$ is $-L/2 \leq m, n \leq L/2$. For the loop count $L_{m,n}$, first note that the period of any such loop after a ket site comes back to another ket site after the contractions is $p = |m-n|$. Furthermore, the available region for the loop is $(A \cup (A+p))^c$, for two sites separated by $p$, neither one can be inside region $A$. Thus, the number of loops $L_{m,n}$ is upper bounded by (this is analogous to Eq.~\eqref{eq:number_cycles})
    \begin{equation}
        L_{m,n} \leq \frac{L - |A \cup (A + p)|}{L/p} = p \left( 1 - \frac{|A| + \min(|A|,p)}{L}\right).
    \end{equation}
    Since $p = |m-n|$ and $|A| \leq L/2$, one can maximize the combination $L_{m,n} + V_{m,n}/2$, with $V_{n,m} \to \min(|m|, |n|, |A|)$, over the allowed values of $m$ and $n$ to find that $L_{m,n} + V_{m,n}/2 \leq L - 3 |A|/2$. Putting this bound back in the trace norm, we have
    \begin{align}
        \norm{R_{m,n}}_1  \leq d^{-(L - |A|) + L - 3|A|/2} = d^{-|A|/2}.
    \end{align}
\end{proof}

As an immediate corollary, we can restrict to the diagonal terms $R_n$ up to an exponentially small factor:
\begin{corollary}
    $R_A = \frac{1}{L} \sum_{n=0}^{L-1} R_n + O(L^2d^{-|A|/2})$, for $C > 0$ a constant.
\end{corollary}

More precisely, to be able to neglect the error term above, we will henceforth assume that $|A| \gg \log L$. For simplicity, we will assume that the equality above holds exactly.

For the approximate orthogonality of $R_n$, we separate the decomposition $R_A = \sum_{n=0}^{L-1} R_n = \sum_{n=-\lceil L/2 \rceil}^{\lfloor L/2 \rfloor} R_n$ into one sum of terms from $n = -|A|+1$ to $|A|-1$, and another with the rest of the terms. We do this because for $|A| \leq |n| \leq L/2$, $R_n$ is the maximally mixed state, $R_n \propto \one_{A \bar{A}}$. In that way, we have $R_A = \sum_{n=-|A|+1}^{|A|} p_n R_n$, where $p_n = 1/L$ if $|n| < |A|$ and $p_{|A|} = 1 - \frac{2|A|-1}{L}$. Hence, below, we will argue for the orthogonality between $R_n$ and $R_m$ if  both $n,m \leq |A|$, and between $R_n$ and the maximally mixed state. 

We separate the analysis of the von Neumann and Rényi entropies due to technical differences in their proofs

\subsection{Rényi operator space entanglement}
\label{app-sec:renyi_OSE}
We turn to the Rényi-$k$ operator space entanglement calculation. We will prove results for integer $k \geq 2$, but will later conjecture the value for other real-valued $k$ through analytical continuation. The final results are illustrated in Fig.~\ref{fig:operator_space_entanglement_comparison}.

\begin{figure}
    \centering
    \begin{tikzpicture}

\tikzset{
    r1line/.style={thickline, red},
    r3line/.style={thickline, blue},
    crossr1/.style={r1line, preaction={draw, white, line width=4.5pt}},
    crossr3/.style={r3line, preaction={draw, white, line width=4.5pt}},
    crossblue/.style={thickline, blue, preaction={draw, white, line width=4.5pt}},
    reddash/.style={thickline, red, dashed, dash pattern=on 4pt off 4pt}
}

\def\L{6}           %
\def\Shift{0.4}     %
\def\outX{3.2}      %
\def\outY{2.0}      %
\def\inX{3.2}       %
\def\inY{0}         %

\drawblock{out}{\L}{\Shift}{\outX}{\outY}
\drawblock{in}{\L}{\Shift}{\inX}{\inY}

\coordinate (leftLabel) at (\inX + 0.8, \outY/2 + \inY/2 + \ydistance/2);
\node[left] at (leftLabel) {\Large $\text{Tr}[\textcolor{red}{R_1}\textcolor{blue}{R_3}] = d^{-7} \cdot d^{-9}$};

\coordinate (rightLabel) at (\inX + \L*\xdistance + \Shift + 0.5, \outY/2 + \inY/2 + \ydistance/2);
\node[right] at (rightLabel) {\Large $= d^{-2 \cdot 6}$};

\connectIntra[r3line]{in}{1}{4}
\connectIntra[r3line]{in}{2}{5}
\connectIntra[r3line]{in}{3}{6}

\connectIntra[crossr1]{in}{1}{2}
\connectIntra[crossr1]{in}{2}{3}
\connectIntra[crossr1]{in}{3}{4}
\connectIntra[crossr1]{in}{4}{5}
\connectIntra[crossr1]{in}{5}{6}

\connectInter[crossblue]{out}{T}{2}{in}{T}{2}
\connectInter[crossblue]{out}{T}{3}{in}{T}{3}
\connectInter[crossblue]{out}{B}{4}{in}{B}{4}
\connectInter[crossblue]{out}{B}{5}{in}{B}{5}

\connectInter[crossblue]{out}{T}{1}{in}{T}{1}
\connectInter[reddash]{out}{T}{1}{in}{T}{1}

\connectInter[crossblue]{out}{B}{6}{in}{B}{6}
\connectInter[reddash]{out}{B}{6}{in}{B}{6}

\connectIntra[crossr3]{out}{1}{4}
\connectIntra[crossr3]{out}{2}{5}
\connectIntra[crossr3]{out}{3}{6}

\connectIntra[crossr1]{out}{1}{2}
\connectIntra[crossr1]{out}{2}{3}
\connectIntra[crossr1]{out}{3}{4}
\connectIntra[crossr1]{out}{4}{5}
\connectIntra[crossr1]{out}{5}{6}

\end{tikzpicture}
    \caption{Accompanying illustration to the proof of Lemma \ref{lemma:renyi-2_overlap}. Depicted here is the Hilbert-Schmidt inner product $\Tr[R_1 R_3]$ for $|A| = 6$. This trace is equivalent to superposing the string diagrams for both $R_1$ and $R_3$ and counting $d^{L_{1,3}}$, where $L_{1,3} = 4$ is the number of loops formed after contraction. Here, we can see that there is one loop for each pair of vertical lines from both $R_1$ and $R_3$, including coincidental ones represented by the red-and-blue dashed lines.}
    \label{fig:hs_overlap}
\end{figure}

First, we prove that the Hilbert-Schmidt overlap between the $R_n$ states is exponentially small in $|A|$. This measure of orthogonality will suffice for the Rényi-$k$ case with $k\geq 2$, since it is a native Rényi-2-type quantity.
\begin{lemma}\label{lemma:renyi-2_overlap}
    $\Tr[R_m R_n] = d^{-2|A|}$ for all $m \neq n$.
\end{lemma}
\begin{proof}
    For simplicity, we will only treat the case where $m$ and $n$ have the same sign, which we assume to be positive without loss of generality. See Fig. \ref{fig:hs_overlap} for an illustrative example of the arguments below.

    Each state $R_n$ consists of $2n$ vertical lines corresponding to identity operators, split half and half across different sides of the region $A$ and between ket and bra spaces. Because of this, any loop formed in the contraction of the trace $\Tr[R_m R_n]$, a $\min(|m|, |n|)$ has to pass through two of them, and only two. Thus, the number of loops is $L_{m,n} = |m|+|n|$, implying
    \begin{align}
        \Tr[R_m R_n] = d^{-(|A| + |m||)} d^{-(|A|+|n|)} d^{|m| + |n|} = d^{-2|A|}.
    \end{align}
\end{proof}

\begin{theorem}\label{thm:renyi_ose_integer}
    For $k \in \Z$, $k \geq 2$,
    $S_k(R_A) = \frac{k}{k-1} \log L - \frac{1}{k-1} \log \left(\frac{1 + d^{-2(k-1)}}{1-d^{-2(k-1)}}\right) + O(L^{2k} d^{-2|A|(k-1)})$
\end{theorem}
\begin{proof}
    Given the reduced decomposition of $R_A = \sum_{n=-|A|+1}^{|A|} p_n R_n$, the trace $\Tr[R_A^k]$ will separate into $2|A|$ terms coming from the diagonal $\Tr[R_n^k]$, and $(2|A|)^k-2|A|$ terms from the off-diagonal $\Tr[R_{n_1} \cdots R_{n_k}]$, with at least one pair of distinct terms $n_i \neq n_j$. Using \cref{lemma:renyi-2_overlap}, we can bound any such off-diagonal term by
    \begin{equation}
        |\Tr[R_{n_1} \cdots R_{n_k}]| \leq \norm{\prod_{l \neq i,j} R_{n_l}}_\infty \Tr[R_{n_i} R_{n_j}] \leq d^{-2|A|}.
    \end{equation}
    Hence,
    \begin{align}
        \Tr[R_A^k] & = \sum_{n_1, \ldots, n_k} \left(\prod_{i=1}^k p_{n_i}\right)\Tr[R_{n_1}\cdots R_{n_k}] \\
        & = \sum_n p_n^k\Tr[R_n^k] +O(|A|^kd^{-2|A|})\\
        & = \sum_{n=-|A|+1}^{|A|-1} \frac{1}{L^k} d^{-2|n| (k-1)} + \left(1 - \frac{2|A|-1}{L}\right)^k d^{-2|A|(k-1)} + O(|A|^kd^{-2|A|}) \label{eq:renyi_for_analytical_continuation} \\
        & = \frac{1}{L^k} \frac{1 + d^{-2(k-1)}}{1-d^{-2(k-1)}} + O(|A|^kd^{-2|A|(k-1)}).
    \end{align}
    By taking the log of this expression for the Rényi-$k$ entropy, we arrive at
    \begin{align}
        S_k(R_A) & = -\frac{1}{k-1} \log \Tr[R_A^k] \\
        & = \frac{k}{k-1} \log L - \frac{1}{k-1} \log \left(\frac{1 + d^{-2(k-1)}}{1-d^{-2(k-1)}}\right) + O(L^{2k} d^{-2|A|(k-1)}).
    \end{align}
\end{proof}

From this result, we see that the Rényi-$k$ entropy only depends on the total system size $L$, and not on the subregion size $|A|$, up to an approximation error $O(L^k d^{-2|A|(k-1)})$ that becomes relevant only for region sizes $|A| \lessapprox \log L$.

Although the final expression in the theorem above was calculated assuming integer $k \geq 2$, we can consider any real $k \in \R$, starting from \cref{eq:renyi_for_analytical_continuation}.

\begin{figure}[t]
    \centering
\begin{tikzpicture}[scale=0.7]
    \begin{axis}[
        width=10cm, height=7cm,
        axis lines=left,
        xlabel={\Large $|A|$},
        ylabel={\Large $S_k(R_A)$},
        every axis x label/.style={at={(ticklabel* cs:1)}, anchor=west},
        every axis y label/.style={at={(ticklabel* cs:1)}, anchor=south},
        xmin=0, xmax=1.1,
        ymin=0, ymax=1.15,
        xtick={0.5, 1},
        xticklabels={\large $L/2$,\large $L$},
        ytick={0.16, 0.5, 1.0},
        yticklabels={\large $O(\log L)$,\large $L/2$,\large $L$},
        tick style={line width=1.5pt, color=black},
        clip=false,
        axis line style={-{Stealth[length=10pt]}, line width=1.5pt},
        axis on top
    ]

        \addplot[domain=0:0.5, dotted, line width=1.5pt] {2*x};
        \addplot[domain=0.5:1, dotted, line width=1.5pt] {2*(1-x)};

        \draw[dashed, line width=1.5pt] (axis cs:0, 1) -- (axis cs:0.5, 1);
        \addplot[domain=0:1, samples=200, color=green!60!black, line width=3pt] {1.001-sqrt(0.01+4*(x-1/2)^2)};

        \draw[dashed, line width=1.5pt] (axis cs:0, 0.5) -- (axis cs:0.5, 0.5);
        \addplot[domain=0:1, samples=200, color=blue, line width=3pt] {2*x*(1-x)};

        \draw[dashed, line width=1.5pt] (axis cs:0, 0.16) -- (axis cs:0.2, 0.16);
        \addplot[domain=0:1, samples=200, color=red, line width=3pt] {0.16 * (1 - abs(2*x-1)^6)};

        \node[text=green!60!black, font=\Large\bfseries] at (axis cs:0.65, 0.94) {$k < 1$};
        \node[text=blue, font=\Large\bfseries] at (axis cs:0.5, 0.58) {$k = 1$};
        \node[text=red, font=\Large\bfseries] at (axis cs:0.5, 0.24) {$k > 1$};

    \end{axis}
\end{tikzpicture}
    \caption{Rényi-$k$ operator space entanglement profiles. The dotted line indicates the maximum entanglement over all states, $S_k \leq 2\min(|A|, L - |A|)  \log d $, while the dashed lines indicate the approximate maxima achieved for each $S_k(R_A)$. For integer $k \geq 2$, it was shown in Theorem~\ref{thm:renyi_ose_integer} to be constant equal to $\frac{k}{1-k} \log L + O(1)$ if $|A| \gg \log L$. The von Neumann operator space entanglement corresponds to $k=1$, and was shown in Theorem~\ref{thm:von_Neumann_ose} to follow a volume-law $S(R_A) = 2|A| (1-\frac{2|A|}{L}) \log d + O(\log L)$. For the rest ($0 \leq k < 1$ and non-integer $k > 1$), we display results after analytical continuation from $k \in \Z$, $k \geq 2$ (Corollary \ref{corollary:renyi_ose_analytical_cont}). In particular, for $0 \leq k < 1$, we have $S_k = 2 \min(|A|, L - |A|)  \log d + O(1)$.}
    \label{fig:operator_space_entanglement_comparison}
\end{figure}
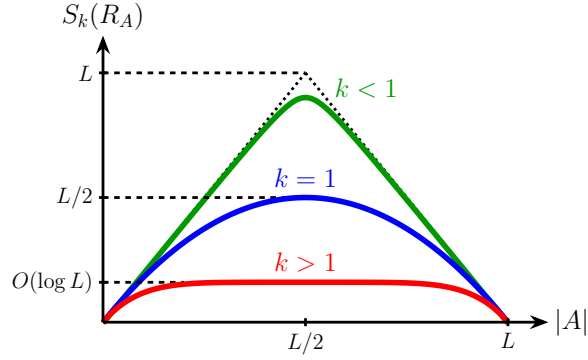

\begin{corollary}\label{corollary:renyi_ose_analytical_cont}
    By analytically continuing the results for $S_k(R_A)$ from $k \in \Z$ and $k\geq 2$ to $k \in \R_+$, and for $|A| \gg \log L$, we have
    \begin{equation}
        S_k(R_A) = \begin{cases}
            \frac{k}{k-1} \log L  , & \text{if } k > 1 \\
            2 |A|\left(1-\frac{|A|}{L}\right) \log d + 2\min\left(\frac{|A|}{L}, 1 - \frac{|A|}{L}\right)\log L, & \text{if } k=1 \\
            2 \min(|A|, L - |A|) \log d, & \text{if } 0 \leq k < 1
        \end{cases}
        + O(1),
    \end{equation}
\end{corollary}
\begin{proof}
    For $k > 1$, it is easy to see that the exact same result reported integer $k$ holds. For $0 \leq k < 1$, the summation terms with large $|n|$ and the second term in the expression of \cref{eq:renyi_for_analytical_continuation} dominate, and we have (for $|A| \leq L/2$)
    \begin{equation}
        \Tr[R^k_A] = \left( \frac{1}{L^k} \frac{1 + d^{2(1-k)}}{d^{2(1-k)}-1} + \left(1 - \frac{2|A|-1}{L}\right)^k \right) d^{2|A|(1-k)} + O(L^{-k}),
    \end{equation}
    which after taking the log and discarding a constant piece, we arrive the stated scaling.

    For the $k=1$ limit, we have
    \begin{align}
        \lim_{k\to 1} S_k & = \lim_{k \to 1} - \frac{1}{k-1} \log \Tr[R^k_A] \\
        & = - \frac{d}{dk} \Tr[R^k_A] \bigg|_{k=1} \\
        & = - \sum_{n=|A|+1}^{|A|-1} \log\left(\frac{d^{-2|n|}}{L}\right) \frac{1}{L} - \log \left[ \left( 1 - \frac{2|A|-1}{L}\right)d^{-2|A|}\right] \left( 1 - \frac{2|A|-1}{L}\right) + O(|A|  d^{-2|A|}),
    \end{align}
    which after simple algebraic manipulations, reduces to the stated result. 
\end{proof}

In the next section, we confirm the scaling from the analytical continuation procedure above in the von Neumann entropy case by employing more sophisticated methods of using the approximate orthogonality of $R_n$. 

\subsection{Von Neumann operator space entanglement}\label{app-sec:von_Neumann_OSE}

We will follow a similar proof strategy as shown in the previous section, for the Rényi-$k$ operator space entanglement, with the crucial difference lying on the choice of orthogonality measure. Before, it was the Hilbert-Schmidt inner product. Here, since the von Neumann entropy is a Rényi-1-type quantity, it requires an orthogonality measure of similar caliber, which is the quantum fidelity between states: 
\begin{lemma}\label{lemma:fidelity_translation_states}
    $F(R_m, R_n) = d^{-|A| + \min(|A|,|m|,|n|, p_{m,n})} \leq d^{-|A| + \min(|A|,|m|, |n|)}$, where $p_{m,n} = \min(|A|\mod{|m-n|}, |m-n| - (|A| \mod{|m-n|}))$
\end{lemma}

The expression above can be analyzed in the same way as the ones in Lemmas \ref{lemma:off_diagonal_vanish} and \ref{lemma:renyi-2_overlap}. For that reason, we will skip its proof, but note that the expression noted above was numerically verified for a large range of parameters $m$, $n$ and $|A|$.

\begin{theorem}\label{thm:von_Neumann_ose}
    $S(R_A) = 2 |A|\left(1-\frac{|A|}{L}\right) \log d+ \frac{2|A|}{L}\log L + O(1)$
\end{theorem}
\begin{proof}
    Let $\tilde{R}_A$ be the flagged version of the decomposition of $R_A$ into $R_n$: $\tilde{R}_A \defeq \sum_{n=-|A|+1}^{|A|} p_n R_n \otimes \ketbra{n}{n}_F$, where $F$ stands for the flag Hilbert space. Then, 
    \begin{align}
        S(F|A)_{\tilde{R}_A} & = S(AF)_{\tilde{R}_A} - S(A)_{\tilde{R}_A} \\
        & = H(p) + \sum_{n=-|A|+1}^{|A|} p_n S(R_n) - S(R_A).
    \end{align}
    The first two terms are the approximate expression for $S(R_A)$, since
    \begin{align}
        H(p) & = -\sum_{n=-|A|+1}^{|A|-1} \frac{1}{L} \log \frac{1}{L} - \left(1 - \frac{2|A|-1}{L}\right) \log \left(1 - \frac{2|A|-1}{L}\right) \\
        & = \frac{2|A|}{L} \log L + O(1), \\
        \sum_{n=-|A|+1}^{|A|} p_n S(R_n) & = \sum_{n=-|A|+1}^{|A|-1} \frac{1}{L} (2|n| \log d) + \left(1 - \frac{2|A|-1}{L}\right) (2|A| \log d) \\
        & = 2 |A| \left( 1 - \frac{|A|}{L}\right) \log d.
    \end{align}
    Hence, we just need to bound the conditional entropy $S(F|A)_{\tilde{R}_A}$. To do that, we first convert the quantum system $A$ into a classical one $A_C$ by applying a measurement POVM $\mathcal{E}$ that attempts to distinguish the states $R_n$. Later, we will specify the measurement. Then, from the data processing inequality $S(F|A) \leq H(F|A_C)$, it suffices to bound the classical relative entropy $H(F|A_C)$. For that, we use Fano's inequality
    \begin{equation}
        H(F|A_C) \leq H(p_e) + p_e \log(2|A| - 1),
    \end{equation}
    where $p_e = P(F \neq A_C)$ is the binary probability that the measurement $\mathcal{E}$ fails to measure the correct state. If we choose $\mathcal{E}$ to be the ``pretty good measurement'', one can bound $p_e$ by a function of the fidelities $F(R_n, R_m)$ as such~\cite{barnum_reversing_2002}
    \begin{equation}
        p_e \leq \sum_{m \neq n} \sqrt{p_m p_n} F(R_m, R_n).
    \end{equation}
    Plugging the result of \cref{lemma:fidelity_translation_states} into the expression above and performing the sum, we have 
    \begin{align}
        p_e & \leq \sum_{\substack{|m|,|n| < |A| \\ m\neq n}} \frac{1}{L} d^{-|A| + \min(|m|, |n|)} + 2 \sum_{m = -|A|+1}^{|A|-1} \sqrt{\frac{1}{L}\left(1 - \frac{2|A|-1}{L} \right)} d^{-|A| + |m|} \\
         & = O(1/L) + O(1/\sqrt{L}).
    \end{align}
    Thus, $p_e \leq O(1/\sqrt{L})$, which in turn implies $S(F|A) \leq H(F|A_C) \leq O(\log L / \sqrt{L})$.
\end{proof}

It is instructive to compare the scaling behavior of the operator space entanglement of $\rho_\tau$ with the entanglement of the Motzkin chain, due to their similarities. At half-chain $|A| = L/2$, their qualitative scaling behaviour coincide for $\alpha \neq 1$, but disagree at $\alpha=1$, with the von Neumann entropy of the Motzkin chain growing slower as $\sqrt{L}$~\cite{sugino_renyi_2018, menon_symmetries_2024}.

\section{Non-linear correlators, doubled states, and SW-SSB}\label{appendix:main_swssb}

In this section, we derive certain non-linear observables of the MMIS $\rho_\tsym$, and discuss the corresponding implications to (i) the properties of the doubled state $\kett{\rho_\tsym}$ defined in the enlarged Hilbert space $\mathcal{H}^{\otimes L} \otimes \mathcal{\overline{H}}^{\otimes L}$ (see Appendix.\ref{sec:renyi_2_local}), and (ii) the strong-weak spontaneous symmetry breaking of $\rho_\tsym$ (see Appendix.\ref{appendix:swssb}).  

\subsection{Renyi-2 correlation of local operators}\label{sec:renyi_2_local}

Given a density matrix $\rho =  \sum_{\boldsymbol{\sigma}}  \rho \ket{\boldsymbol{\sigma}} \bra{\boldsymbol{\sigma}}$ acting on the $L$-qudit Hilbert space  $\mathcal{H}^{\otimes L}$, where $\ket{\boldsymbol{\sigma}}$ is a computational basis product state of qudits, we can define the (un-normalized) doubled state $\kett{\rho}$ in the enlarged Hilbert space $\mathcal{H}^{\otimes L} \otimes \mathcal{\overline{H}}^{\otimes L}$: 

\begin{equation}
    \kett{\rho} =  \sum_{\boldsymbol{\sigma}}  \rho \ket{\boldsymbol{\sigma}} \ket{\boldsymbol{\sigma}},  
\end{equation}
where $\rho$ only acts on the first copy (i.e., the Hilbert space of $\mathcal{H}^{\otimes L }$).

When choosing $\rho$ as the MMIS $\rho_\tsym = \frac{1}{\mathcal{N}} \sum_{p=0}^{L-1} T^{p} \propto P_{\tsym}$, where $\mathcal{N}$ is a normalization constant, the associated doubled state $\kett{\rho_\tsym}$ defined in the enlarged Hilbert space $\mathcal{H}^{\otimes L} \otimes \mathcal{\overline{H}}^{\otimes L}$ is: 

\begin{equation}
\kett{\rho_{\tsym}}  \propto \sum_{\boldsymbol{\sigma}} P_\tsym \ket{\boldsymbol{\sigma}} \ket{\boldsymbol{\sigma}} \propto  \sum_{\boldsymbol{\sigma}} (\one+T + ...+ T^{L-1})\ket{\boldsymbol{\sigma}} \ket{\boldsymbol{\sigma}}.
\end{equation}

In the following, we will consider the case of qubits ($d=2$), but the calculation can be generalized to any qudit dimensions as well. Also, for simplicity, we consider the case where $L$ is prime. We will derive the connected correlation function of the doubled state, defined as

\begin{equation}
R^{(2)}_{ij}  = \bbra{\rho_\tsym} Z_i \overline{Z}_i  Z_j \overline{Z}_j \kett{\rho_\tsym}  - \bbra{\rho_\tsym} Z_i \overline{Z}_i \kett{\rho_\tsym}   \bbra{\rho_\tsym} Z_j \overline{Z}_j \kett{\rho_\tsym}, 
\end{equation}
where $Z_i $ and $ \overline{Z}_i$ denote the Pauli-Z acting on $i$-th site of the $\mathcal{H}^{\otimes L}$ and $\mathcal{\overline{H}}^{\otimes L}$, respectively, and $\kett{\rho_\tsym}$ has been appropriately normalized.

We will show that there is a long-range order between the two local operators $Z_i  \overline{Z}_i$ and  $Z_j  \overline{Z}_j$: the correlation does not decay exponentially with their separation, albeit the $\frac{1}{L}$ decay in magnitude. This is already sufficient to prove the long-range entanglement of the doubled state $\kett{\rho_\tsym}$. Namely, $\kett{\rho_\tsym}$ cannot be prepared by a finite-depth local unitary circuit starting from a product state.

We also note that the above correlation function of the doubled state is exactly equal to the connected Rényi-2 correlator, defined as

\begin{equation}
R^{(2)}_{ij} = \frac{ \Tr \rho_\tsym Z_i Z_j \rho_\tsym Z_iZ_j  }{\Tr \rho_\tsym ^2  } - \frac{ \Tr \rho_\tsym Z_i \rho_\tsym Z_i   }{\Tr \rho_\tsym ^2  }  \frac{ \Tr \rho_\tsym Z_j \rho_\tsym Z_j   }{\Tr \rho_\tsym ^2  }.
\end{equation} 
Hence, we will perform the calculation based on this expression. 

As a warm-up, we first calculate  $\Tr \rho_\tsym Z_i \rho_\tsym Z_i$: 

\begin{equation}
\begin{split}
\Tr \rho_\tsym Z_i \rho_\tsym Z_i  &=  \frac{1}{\mathcal{N}^2}  \sum_{p,q =0}^{L-1} \Tr( T^{p} Z_i   T^{q} Z_i )\\
& = \frac{1}{\mathcal{N}^2}  \sum_{p,q =0}^{L-1} \Tr( T^{p}    T^{q} Z_{i-q} Z_i )\\ 
& = \frac{1}{\mathcal{N}^2}  \sum_{p,q =0}^{L-1} \Tr( T^{p+q}   Z_{i-q} Z_i )\\
\end{split}
\end{equation}
When $q+p \neq  0 ~\text{mod}~ L$, one has 
\begin{equation}
\Tr( T^{p+q}   Z_{i-q} Z_i )=  \Tr( T^{q+p} ) = 2. 
\end{equation} 
For the first equality, we have used the ergodicity of the $(q+p)$-shift (because $L$ is prime) and the cyclic property of the trace to move the pairs of Pauli-Zs to a single site for cancellation, as in Eq.\ref{eq:cyclic_trace_2} for the two-point function calculation. In the second equality, we expand the trace in the computational basis $\boldsymbol{\sigma}$, and notice that there are only two possible $\boldsymbol{\sigma}$s that can contribute: $\sigma_i=0 ~\forall ~i$ and  $\sigma_i=1 ~\forall ~i$.  

When $q+p = 0 ~\text{mod}~ L $, one has

\begin{equation}
\Tr( T^{p+q}   Z_{i-q} Z_i ) = \Tr(Z_{i-q} Z_i )= \delta_{q,0} ~2^L.
\end{equation}

Combining the results for both $p+q \neq  0$ and $p+q =0$, one finds 

\begin{equation}\label{eq:rho_zi_rho_zi}
\begin{split}
\Tr \rho_\tsym Z_i \rho_\tsym Z_i= \frac{1}{\mathcal{N}^2} \left[   2^L  + 2L(L-1)      \right] =  \frac{1}{\mathcal{N}^2} \left[   2^L  + O(L^2)      \right].
\end{split}
\end{equation}

On the other hand,
\begin{equation}\label{eq:rho_square}
\begin{split}
\Tr \rho_\tsym^2  &= \frac{1}{\mathcal{N}^2} \sum_{p,q =0}^{L-1} \Tr( T^{p+q}   )  \\
&= \frac{1}{\mathcal{N}^2}  \left[ \sum_{p+ q =0}^{L-1} \Tr( T^{p+q}   ) +\sum_{p+ q \neq 0}^{L-1} \Tr( T^{p+q}   )   \right]\\
&= \frac{1}{\mathcal{N}^2}    \left[ L 2^L + 2L(L-1 )  \right]\\
&=  \frac{1}{\mathcal{N}^2}    \left[ L 2^L + O(L^2) \right]
\end{split}
\end{equation}

By combining Eq.\ref{eq:rho_zi_rho_zi} and Eq.\ref{eq:rho_square}, one finds 

\begin{equation}
\frac{ \Tr \rho_\tsym Z_i \rho_\tsym Z_i   }{\Tr \rho_\tsym ^2  } = \frac{1}{L} + O(2^{-L} L). 
\end{equation}

The quantity $\frac{ \Tr \rho_\tsym Z_i Z_j \rho_\tsym Z_iZ_j  }{\Tr \rho_\tsym ^2  }$ can be calculated using a similar strategy:

\begin{equation}
    \begin{split}
\Tr \rho_\tsym Z_i Z_j \rho_\tsym Z_iZ_j &   =  \frac{1}{\mathcal{N}^2}  \sum_{p,q =0}^{L-1} \Tr( T^{p} Z_iZ_j    T^{q} Z_iZ_j  )\\\\
& = \frac{1}{\mathcal{N}^2}  \sum_{p,q =0}^{L-1} \Tr( T^{p+q }   Z_{i-q }Z_{j-q}  Z_iZ_j  )
    \end{split}
\end{equation}
For the summation with $p+q \neq 0$, one again finds 

\begin{equation}
     \sum_{p + q  \neq 0} \Tr( T^{p+q}  Z_{i-q }Z_{j-q}  Z_iZ_j  ) =    \sum_{p + q  \neq 0} \Tr( T^{p+q}  ) =  2 L (L-1). 
\end{equation}

For the summation with $p+q = 0$, one finds 

\begin{equation}
     \sum_{p+ q =0} \Tr( T^{p+q }   Z_{i-q }Z_{j-q}  Z_iZ_j  )  =    \sum_{p+ q =0} \Tr(Z_{i-q }Z_{j-q}  Z_iZ_j  )
\end{equation}
For $i\neq j$, only the term corresponding to $p=q=0$ can give a non-vanishing value, and therefore, 
\begin{equation}
     \sum_{p+ q =0} \Tr( T^{p+q }   Z_{i-q }Z_{j-q}  Z_iZ_j  ) = 2^L
\end{equation}

Combining the results for both $p+q \neq  0$ and $p+q =0$, one finds

\begin{equation}\label{eq:rho_ZZ_rho_ZZ}
    \begin{split}
\Tr \rho_\tsym Z_i Z_j \rho_\tsym Z_iZ_j & = \frac{1}{\mathcal{N}^2} \left[ 2^L +  2L (L-1) \right].
    \end{split}
\end{equation}

With the results for both $\Tr \rho_\tsym^2 $ and $\Tr \rho_\tsym Z_i Z_j \rho_\tsym Z_iZ_j$ from Eq.\ref{eq:rho_square} and Eq.\ref{eq:rho_ZZ_rho_ZZ}, one finds 
\begin{equation}
\frac{ \Tr \rho_\tsym Z_i Z_j \rho_\tsym Z_iZ_j  }{\Tr \rho_\tsym ^2  }  =   \frac{1}{L} +  O(2^{-L} L). 
\end{equation}

Finally, the connected Rényi-2 correlator is 

\begin{equation}
R^{(2)}_{ij} =  \frac{1}{L} -  \frac{1}{L^2} + \cdots, 
\end{equation}
where $\cdots$ denotes the contribution that vanishes exponentially in $L$. In the leading order, $R^{(2)}_{ij}$ scales as $\frac{1}{L}$.

\subsection{SWSSB of translation}\label{appendix:swssb}
Here we will show that the MMIS $\rho_\tsym$ exhibits a strong-to-weak spontaneous symmetry breaking (SWSSB) of the lattice translation. To this end, below we will first explain the physical mechanism of SWSSB, and then discuss the diagnostic. 

Heuristically, SWSSB refers to a phenomenon where a strongly symmetric mixed state exhibits no long-range order, but develops a long-range order in the pure states that form an ensemble of the mixed state \cite{lee_2023_swssb,you_2024_swssb,lessa_2025_swssb}. Correspondingly, the probe of such order typically relies on certain order parameters non-linear in the density matrix. 

One simple way to think about SWSSB is to consider the double-state representation of the mixed state $\rho_\tsym$:

\begin{equation}
\kett{\rho_{\tsym}}  \propto \sum_{\boldsymbol{\sigma}} P_\tsym \ket{\boldsymbol{\sigma}} \ket{\boldsymbol{\sigma}} \propto  \sum_{\boldsymbol{\sigma}} (\one+T + ...+ T^{L-1})\ket{\boldsymbol{\sigma}} \ket{\boldsymbol{\sigma}}.
\end{equation}
defined in the doubled Hilbert space $\mathcal{H}^{\otimes L} \otimes \mathcal{\overline{H}}^{\otimes L}$. Under such representation, the original strong symmetry of $\rho_\tsym$, i.e. $T \rho_\tsym = \rho_\tsym  =  \rho_\tsym T^{\dagger}$, translates to $T\times \overline{T}$ symmetry of the doubled state $\kett{ \rho_\tsym  }$, namely, $\kett{ \rho_\tsym  } =  T  \kett{ \rho_\tsym  }  = \overline{T} \kett{ \rho_\tsym  }$, where $T$ and $\overline{T}$ act as the translations in $\mathcal{H}^{\otimes L}$ and $\mathcal{\overline{H}}^{\otimes L}$, respectively. 

The strong symmetry $T\times \overline{T}$ can spontaneously be broken down to the diagonal symmetry generated by $T\overline{T}$, which corresponds to the weak symmetry of $\rho_{\tsym}$ and manifests as $T\rho_\tsym T^{\dagger} = \rho$. If this occurs, we say $\rho_\tsym$ exhibits the strong-to-weak spontaneous symmetry breaking (SWSSB) of the translation symmetry. To probe the symmetry-breaking pattern, as in the conventional SSB, we investigate whether there is a corresponding operator charged under the symmetry that exhibits a long-range order or correlation. To probe the SWSSB in the doubled-state representation, we investigate the correlation of operators that are charged under the individual translation $T$ and $\overline{T}$, but are neutral under $T\overline{T}$. One such operator is $O_k \overline{O}_{-k}$ defined as

\begin{equation}
    \begin{split}
    &O_k = \frac{1}{L } \sum_{x=0}^{L-1} e^{ikx }Z_x,\\
    &\overline{O}_k = \frac{1}{L } \sum_{x=0}^{L-1} e^{ikx }\overline{Z}_x,
    \end{split}
\end{equation}
where $x=\{0,1,2,..., L-1 \}$ labels the real-space position, and $k\in \{ \frac{2\pi n }{L } | n=0,1,2, ... , L-1   \}$ labels the momenta. Since $TO_kT^{-1}  =  e^{-ik } O_k$ and $\overline{T}\overline{O}_{-k}\overline{T}^{-1}  =  e^{ik } \overline{O}_{-k}$, $O_k$ and $\overline{O}_{-k}$ are individually charged under $T$ and $\overline{T}$ as long as $k\neq 0 ~\text{mod} ~ 2\pi$, but are neutral under $T\overline{T}$ due to the cancellation of the phase. 

With this operator, one can then diagnose SWSSB using the so-called Rényi-2 correlator, which is also equal to the Rényi-1 correlator since the mixed state $\rho_\tsym$ is a projector. As we will show below in Appendix. \ref{sec:conventional_renyi_2}, such a correlation exhibits $\frac{1}{L^2}$ decays in the leading order, which does not conclusively identify SWSSB. In Appendix.\ref{sec:new_renyi_2}, we introduce the `variance normalized' Rényi-2 correlator, which is more appropriate due to the non-unitary nature of the charged operators. We analytically show that this newly defined quantity is a non-zero constant of the system size $L$, thereby indicating the SWSSB of translation. 

\subsubsection{Rényi-2 correlations}\label{sec:conventional_renyi_2}
With the charged operator defined above, we can explore SWSSB by studying the following correlation function defined for the doubled state:

\begin{equation}
\begin{split}
R^{(2)}_{k} & =  \frac{\bbra{\rho_\tsym} O_k \overline{O}_{-k}  O_{-k} \overline{O}_{k}     \kett{\rho_\tsym} }{\bbrakett{\rho_\tsym| \rho_\tsym  } }\\
\end{split}
\end{equation}
where we have considered $k\neq 0 ~\text{mod} ~ 2\pi$. In terms of the physical mixed state $\rho_\tsym$, the above can also be expressed in terms of 

\begin{equation}
R^{(2)}_{k}  = \frac{ \Tr \rho_\tsym O_k O_{-k} \rho_\tsym  O_k O_{-k}}{  \Tr \rho_\tsym^2 },  
\end{equation}
which is also dubbed the Rényi-2 correlator.

We begin by calculating the numerator:

\begin{equation} \label{eq:okok_numerator}
    \begin{split}
    \Tr \rho O_k O_{-k} \rho  O_k O_{-k}  &= \frac{1}{\mathcal{N}^2} \frac{1}{L^4} \sum_{x,y,a,b}  \sum_{p,q}  e^{ik(x-y) + i k (a-b) }   \Tr \left( T^p Z_x Z_y T^q Z_a Z_b \right) \\
    &= \frac{1}{\mathcal{N}^2} \frac{1}{L^4} \sum_{x,y,a,b}  \sum_{p,q}  e^{ik(x-y) + i k (a-b) }   \Tr \left( T^{p+q}  Z_{x-q} Z_{y-q}  Z_a Z_b \right) \\
    &= \frac{1}{\mathcal{N}^2} \frac{1}{L^4} \sum_{x,y,a,b}  \sum_{p,q}  e^{ik(x-y) + i k (a-b) }   \Tr \left( T^{p+q}  Z_{x} Z_{y}  Z_a Z_b \right) \\
        &= \frac{1}{\mathcal{N}^2} \frac{1}{L^3} \sum_{x,y,a,b}  \sum_{p}  e^{ik(x-y) + i k (a-b) }   \Tr \left( T^{p}  Z_{x} Z_{y}  Z_a Z_b \right) 
     \end{split}
\end{equation}
where $x,y,a,b$ denote the position labels arising from expanding $O_k, O_{-k}$, and $p,q$ denote the position labels arising from expanding $\rho_\tsym$. In the last step, we have redefined the variable $p \to p-q$, and summed over $q$. To proceed, one can divide the summation $\sum_{p}$ into the cases where $p \neq 0$ and $p=0$.

When $p\neq 0$, one can again use the ergodicity of the $p$-shift and the cyclicity of the trace to find 

\begin{equation}\label{eq:non_zero_p}
\Tr \left( T^{p}  Z_{x} Z_{y}  Z_a Z_b \right) = \Tr ( T^{p}  ) =2,
\end{equation}
as in the calculation discussed in Appendix.\ref{sec:renyi_2_local}. 

On the other hand, for $p=0$, one has 

\begin{equation}\label{eq:delta_four_points}
\Tr \left( T^{p}  Z_{x} Z_{y}  Z_a Z_b \right) =  \Tr \left(   Z_{x} Z_{y}  Z_a Z_b \right) =  2^L \left[\delta_{x,y }  \delta_{a,b}  + \delta_{x,a }  \delta_{y,b} +  \delta_{x,b }  \delta_{y,a} -2 \delta_{x,y }  \delta_{x,a}   \delta_{x,b} \right]. 
\end{equation}
The above follows from the fact that the trace can be non-vanishing (specifically, $2^L$) only when all Pauli-Zs vanish; the first three terms come from different pair-wise pairing of Pauli-Zs, and the last term compensates the over-counting when all indices are the same. It follows that 

\begin{equation}
\begin{split}
    \sum_{x,y,a,b}   e^{ik(x-y) + i k (a-b) }   \Tr \left( Z_{x} Z_{y}  Z_a Z_b \right) &= 2^L \left[ L^2 +  |\sum_x  e^{2ik x}|^2  + L^2 -2L  \right]\\
& =  2^L \left[ 2L^2 +  L^2 (\delta_{k,0} +\delta_{k,\pi})     -2L  \right]\\
\end{split}
\end{equation}

Therefore, in the leading order, 

\begin{equation}
   \Tr \rho O_k O_{-k} \rho  O_k O_{-k} = \frac{1}{\mathcal{N}^2} 2^L  ( \frac{2}{L} +  \frac{1}{L} (\delta_{k,0} +\delta_{k,\pi})  + O(\frac{1}{L^2})). 
\end{equation}
We note that in the above equation, the exponentially small correction from $p\neq0$ (Eq.\ref{eq:non_zero_p}) has been included in $O(\frac{1}{L^2}   )$ as well.

On the other hand, Eq.\ref{eq:rho_square} shows 
\begin{equation}
\begin{split}
\Tr \rho^2 &=  \frac{1}{\mathcal{N}^2}    \left[ L 2^L + O(L^2) \right]. 
\end{split}
\end{equation}

It follows that the Rényi-2 correlator for $k\neq 0$ is

\begin{equation}
\boxed{ R^{(2)}_{k}=\frac{1}{L^2}\left( 2+  \delta_{k,\pi}   \right)  + O( \frac{1}{L^3})   }
\end{equation} 
Therefore, $R^{(2)}_{k} $ exhibits $\frac{1}{L^2}$ scaling in the leading order. As a side note, $\delta_{k,\pi}=0$ since $k=\pi$ is never allowed for a finite prime $L>2$.

\subsubsection{Variance-normalized  Rényi-2 correlations}\label{sec:new_renyi_2}
To motivate the variance-normalized Rényi-2 correlator, we consider an alternative aspect of SSB from the response of inserting charged operators. In particular, we define $\tilde{R}_k^{(2)}$ as the overlap between the normalized doubled state $\frac{\kett{\rho_\tsym}}{  \sqrt{  \bbrakett{   \rho_\tsym |  \rho_\tsym   }  } }$ and the doubled state 
subject to the application of the non-unitary operator $O_k \overline{O}_{-k}  O_{-k} \overline{O}_{k}$, namely, $\frac{O_k \overline{O}_{-k}  O_{-k} \overline{O}_{k}\kett{ \rho_\tsym}}{  \sqrt{  \bbrakett{   \rho_\tsym | (O_k \overline{O}_{-k}  O_{-k} \overline{O}_{k})^2| \rho_\tsym   }  } }$. Specifically, 

\begin{equation}
\tilde{R}_k^{(2)} =  \frac{  \bbra{\rho_\tsym}  O_k \overline{O}_{-k}  O_{-k} \overline{O}_{k}\kett{ \rho_\tsym}   }{  \sqrt{  \bbrakett{   \rho_\tsym |  \rho_\tsym   }     \bbrakett{   \rho_\tsym | (O_k \overline{O}_{-k}  O_{-k} \overline{O}_{k})^2| \rho_\tsym   }  }     }.
\end{equation}

In terms of the original density matrix $\rho_\tsym$, $\tilde{R}_k^{(2)}$ can be expressed as 

\begin{equation}
    \tilde{R}_k^{(2)} =  \frac{  \Tr [   \rho_\tsym O_{k}  O_{-k}  \rho_\tsym O_{k}  O_{-k}     ]  }{ \sqrt{\Tr  [\rho_\tsym^2]  \Tr [ ( O_k O_{-k}   \rho_\tsym  O_k O_{-k}  )^2   ]  }},  
\end{equation}
Note that if the charged operator is unitary, the denominator will be reduced to $  \Tr [\rho_\tsym^2]$, and  $\tilde{R}_k^{(2)}$ becomes the conventional Rényi-2 correlator.

Compared to Appendix. \ref{sec:conventional_renyi_2}, the new quantity we need to consider is $\Tr [ ( O_k O_{-k}   \rho_\tsym  O_k O_{-k}  )^2   ]$, which can be calculated as follows:

\begin{equation}
\begin{split}
&\Tr [ ( O_k O_{-k}   \rho_\tsym  O_k O_{-k}  )^2   ]\\
&=  \Tr [ O_k O_{-k}   \rho_\tsym  O_k O_{-k}  O_k O_{-k}   \rho_\tsym  O_k O_{-k}     ]\\  
&=\frac{1}{L^4\mathcal{N}^2}  \sum_{  \{ x_i, y_i | i=1,2,3,4  \} } \sum_{p,q}  e^{ik(x_1-x_2 +x_3 -x_4   )   }  e^{ik(y_1-y_2 +y_3 -y_4   )   } \Tr [  T^p Z_{x_1} Z_{x_2}   Z_{x_3}   Z_{x_4} T^q Z_{y_1} Z_{y_2}   Z_{y_3}   Z_{y_4}      ]  \\
& = \frac{L}{L^4\mathcal{N}^2}  \sum_{  \{ x_i, y_i | i=1,2,3,4  \} }  \sum_{p}  e^{ik(x_1-x_2 +x_3 -x_4   )   }  e^{ik(y_1-y_2 +y_3 -y_4   )   } \Tr [  T^p Z_{x_1} Z_{x_2}   Z_{x_3}   Z_{x_4}  Z_{y_1} Z_{y_2}   Z_{y_3}   Z_{y_4}    ].  
\end{split}
\end{equation}

For prime $L$, as we have seen in Appendix.\ref{sec:conventional_renyi_2}, the contribution from $p\neq  0$ will be exponentially small compared to the contribution from $p=0$. Hence below we will only focus on the $p=0$ contribution. In this case, the central quantity we need to calculate is $\Tr [ Z_{x_1} Z_{x_2}   Z_{x_3}   Z_{x_4}  Z_{y_1} Z_{y_2}   Z_{y_3}   Z_{y_4} ] $. In order to have a non-vanishing trace (specifically, $2^L$), the 8 Pauli-Zs must be an identity, which can arise from various pair-wise pairing patterns of Pauli-Zs. For instance, one can have the following pairing pattern:

\begin{equation}
\contraction{}{Z_{x_1}}{}{Z_{x_2}}
\contraction{Z_{x_1}Z_{x_2} }{Z_{x_3}}{}{Z_{x_4}}
\contraction{Z_{x_1}Z_{x_2}Z_{x_3}Z_{x_4} }{Z_{y_1}}{}{Z_{y_2}}
\contraction{Z_{x_1}Z_{x_2}Z_{x_3}Z_{x_4}Z_{y_1}Z_{y_2} }{Z_{y_3}}{}{Z_{y_4}}
 Z_{x_1} Z_{x_2}   Z_{x_3}   Z_{x_4}  Z_{y_1} Z_{y_2}   Z_{y_3}   Z_{y_4} =  \delta_{x_1,  x_2 }\delta_{x_3,  x_4 } \delta_{y_1,  y_2 }\delta_{y_3,  y_4 }. 
\end{equation}

As such, 

\begin{equation}
    \Tr [ Z_{x_1} Z_{x_2}   Z_{x_3}   Z_{x_4}  Z_{y_1} Z_{y_2}   Z_{y_3}   Z_{y_4} ] = 2^L \left[  \delta_{x_1,  x_2 }\delta_{x_3,  x_4 } \delta_{y_1,  y_2 }\delta_{y_3,  y_4 }+  \cdots \right],
\end{equation}
where $\cdots$ denotes all the other delta functions associated with various pairing patterns with $+$ sign, and additional terms with $-$ sign that compensate for the over-counting when certain indices are the same; see the last term of Eq.\ref{eq:delta_four_points} for example. In particular, with the delta function condition in the first term, the phase $ e^{ik(x_1-x_2 +x_3 -x_4   )   }  e^{ik(y_1-y_2 +y_3 -y_4   )   } =1$, so 

\begin{equation}
\Tr [ ( O_k O_{-k}   \rho_\tsym  O_k O_{-k}  )^2   ] = \frac{L}{L^4\mathcal{N}^2 } 2^L ( L^5+  \cdots)  = \frac{1}{\mathcal{N}^2 } 2^L  \Theta(L^2),
\end{equation}

Combining with the result in Appendix.\ref{sec:conventional_renyi_2}, we find 

\begin{equation}
\tilde{R}_k^{(2)} =  \frac{  \Tr [   \rho_\tsym O_{k}  O_{-k}  \rho_\tsym O_{k}  O_{-k}     ]  }{ \sqrt{\Tr  [\rho_\tsym^2]  \Tr [ ( O_k O_{-k}   \rho_\tsym  O_k O_{-k}  )^2   ]  }} = \Theta(1).
\end{equation}

\end{document}